\documentclass[a4paper, 11pt]{article}
\usepackage{fullpage}
\usepackage{amsmath,amssymb}
\usepackage{amsthm}
\usepackage{amsfonts}
\usepackage{array}
\usepackage{xspace}
\usepackage{enumerate}
\usepackage[utf8]{inputenc}
\usepackage[unicode]{hyperref}
\usepackage[hypcap]{caption}
\usepackage{tikz}
\usetikzlibrary{shapes}
\usetikzlibrary{calc}

\usepackage{datetime} 

\newtheorem{theorem}{Theorem}[section]
\newtheorem{corollary}[theorem]{Corollary}
\newtheorem{proposition}[theorem]{Proposition}
\newtheorem{lemma}[theorem]{Lemma}
\theoremstyle{definition}
\newtheorem{definition}{Definition}[section]
\theoremstyle{remark}
\newtheorem{remark}{Remark}[section]

\newtheorem{myclaim}[theorem]{Claim}

\newcommand{\displaypunct}[1]{\,\,\text{#1}}

\newcommand{\Neci}{Ne\v{c}iporuk\xspace}

\newcommand{\ISA}{{\ensuremath{\text{\sc ISA}}}}
\newcommand{\ED}{{\ensuremath{\text{\sc ED}}}}
\newcommand{\isakl}{{\ensuremath{\ISA_{k,\ell}}}}
\newcommand{\isakk}{{\ensuremath{\ISA_{k,k}}}}
\newcommand{\isan}{{\ensuremath{\ISA_n}}}

\newcommand{\Nsem}{N_{sem}}
\newcommand{\ceiling}[1]{\left\lceil #1 \right\rceil}
\newcommand{\floor}[1]{\left\lfloor #1 \right\rfloor}

\newcommand{\cointerval}[2]{\left[#1, #2\right[}
\newcommand{\ccinterval}[2]{\left[#1, #2\right]}

\newcommand{\card}[1]{\left| #1 \right|}

\newcommand{\function}[5]
{
    #1\colon
    \begin{array}[t]{@{}r@{\,\,}l@{\,\,}l@{}}
    #2 & \to & #3\\
    #4 & \mapsto & #5
    \end{array}
}

\DeclareMathOperator{\N}{\mathbb{N}}
\DeclareMathOperator{\R}{\mathbb{R}}

\newcommand{\M}{\mathbf{M}}
\DeclareMathOperator{\Xix}{\Gamma}
\DeclareMathOperator{\BF}{\mathbf{L}}
\DeclareMathOperator{\BP}{\mathbf{BP}}
\DeclareMathOperator{\NBP}{\mathbf{NBP}}
\DeclareMathOperator{\PBP}{\oplus\mathbf{BP}}
\DeclareMathOperator{\LNBP}{\mathbf{LNBP}}

\DeclareMathOperator{\bin}{bin}
\DeclareMathOperator{\NeciporukSet}{\mathcal{N}}
\DeclareMathOperator{\NeciporukLB}{N}
\DeclareMathOperator{\LNBF}{\mathbf{LL}}

\DeclareMathOperator{\NL}{\mathsf{NL}}
\DeclareMathOperator{\PL}{\mathsf{\oplus{}L}}
\DeclareMathOperator{\NP}{\mathsf{NP}}
\DeclareMathOperator{\Omicron}{O}

\newcommand\set[1]{\{#1\}}
\newcommand\tuple[1]{\overline{#1}}

\definecolor{ngcolor}{rgb}{0,0.6,0.2}

\definecolor{lscolor}{rgb}{0,0.2,0.6}

\definecolor{donecolor}{rgb}{0,0.9,0.3}

\begin{document}

\title
{Nondeterminism and an abstract formulation of Ne\v{c}iporuk's lower
  bound method}

\author{Paul \textsc{Beame}
\thanks{U. Washington, beame@cs.washington.edu. 
Research supported by NSF grants CCF-1217099 and CCF-1524246.}
\and
Nathan \textsc{Grosshans}
\thanks{U. Montr{\'e}al and ENS Cachan, U. Paris-Saclay,
	nathan.grosshans@lsv.ens-cachan.fr.
	This work was supported by grants from Digiteo France.}
\and
Pierre \textsc{McKenzie}
\thanks{U. Montr{\'e}al and ENS Cachan, U. Paris-Saclay,
	mckenzie@iro.umontreal.ca.
Research supported by NSERC of Canada 
and by the ``Chaire DIGITEO, ENS Cachan - \'{E}cole Polytechnique''.}
\and
Luc \textsc{Segoufin}
\thanks{INRIA and ENS Cachan, U. Paris-Saclay, luc.segoufin@inria.fr.}
}

\date{\today}

\maketitle

\thispagestyle{plain} 

\begin{abstract}
  A formulation of \emph{\Neci's lower bound method} slightly more
  inclusive than the usual complexity-measure-specific formulation is
  presented.  Using this general formulation, limitations to lower
  bounds achievable by the method are obtained for several computation
  models, such as branching programs and Boolean formulas having
  access to a sublinear number of nondeterministic bits.  In
  particular, it is shown that any lower bound achievable by the
  method of \Neci\ for the size of nondeterministic 
  and parity branching programs
  is at most $O(n^{3/2}/\log n)$. 
\end{abstract}

\section{Introduction}
\label{sec:introduction}
Relatively few methods exist to prove complexity lower bounds 
in general non-uniform models of computation.
Fifty years ago, \Neci\ wrote his famous two-page note entitled ``A
boolean function''~\cite{ne66}.  That note contained the first super-linear
lower bounds on the size of Boolean formulas over arbitrary bases 
and the size of contact schemes needed to compute some explicit Boolean
function. 

\Neci's~\cite{ne66} method still yields
the best lower bounds known today for explicit functions in a number of
complexity measures.
In particular, there are explicit functions for which \Neci's method yields
lower bounds of $\Omega(n^2/\log n)$ on formula size over an arbitrary basis, 
$\Omega(n^2/\log^2 n)$ on deterministic branching program size and on
contact scheme size, and $\Omega(n^{3/2}/\log n)$ on
nondeterministic branching program size, 
switching-and-rectifier network size, parity branching program size
and span program size. 
All of these are the best known lower bounds for these complexity measures
for any explicit function.
The first two of these lower bounds are contained in \Neci's original
paper~\cite{ne66}.
Pudlak~\cite{pu87} points out that \Neci's method yields the third 
lower bound for {\em nondeterministic} branching program size,
as well as for switching-and-rectifier network size 
program size~\cite{ra91}; 
Karchmer and Wigderson~\cite{kawi93} point out that the Pudlak's
observation extends to parity branching program size and hence also
applies to span program size.

Two simple explicit functions that yield the lower bounds
mentioned above are the Element Distinctness function and the Indirect
Storage Access function.

\Neci's method relies on counting subfunctions induced on blocks in a partition
of the input variables.
It is natural to try to optimize the use of the method, both in terms of how the
bound depends on the numbers of subfunctions for each block in the
partition and whether there are functions other than Element Distinctness and
Indirect Storage Access for which one can prove stronger lower bounds. 

For formula size over arbitrary bases, it is well known~(see, e.g.,~\cite{we87,sav76})
that $\Theta(n^2/\log n)$ is indeed the best lower bound obtainable by
\Neci's method.
Savage~\cite{sav76} also cites Paterson (unpublished) as
improving the constant factor in the bound.
Similarly, $\Theta(n^2/\log^2 n)$ is the best lower bound obtainable by
\Neci's method for deterministic branching program size,
as noted by Wegener~\cite[p. 422]{we87}, who states the claim with a hint
at its proof.  
Moreover, Alon and Zwick~\cite{alzw89} derived the optimal multiplicative
constant in this lower bound as a function of the number of
subfunctions of $f$ in each block.

Since the first two bounds using \Neci's method are asymptotically the
best possible, it is natural to ask whether the third lower bound also
uses \Neci's method in an optimal way.
Jukna seems to be the only one who has explicitly addressed this question.  
In his discussion~\cite[p. 207]{juk01} of the $\Omega(n^{3/2}/\log n)$
lower bound on span program size due to Karchmer and Wigderson~\cite{kawi93} 
and based on \Neci's method, he states that the method ``cannot lead to much
larger lower bounds'' but does not give more details.

In this paper we give a more precise result, namely that the best bound on
nondeterministic and parity branching program
size obtainable by \Neci's method is indeed $\Theta(n^{3/2}/\log n)$.
This automatically applies to span program size and switching-and-rectifier
network size since these measures are upper-bounded by parity and
nondeterministic branching program size, respectively.

In deriving lower bounds using \Neci's method as discussed in the literature,
the major difference between measures is the estimate of the number of functions
of low complexity with respect to each measure.   In most presentations of
\Neci's method, such as those
in~\cite{we87,bosi90,we00}, this bound is determined {\em syntactically},
for example, via a count of the
number of syntactically distinct formulas of a given size
or of syntactically distinct branching programs with a given number of nodes.
However, this is an overcount of the number of different functions since,
for example, many
syntactically distinct branching programs may compute the same function.
If only {\em semantically} distinct objects are counted, one may, in
principle, obtain stronger lower bounds using \Neci's method.   
In the case of deterministic branching programs,
Alon and Zwick~\cite{alzw89} considered the stronger semantic version of
\Neci's bound and showed nonetheless that both semantic and syntactic versions
reach the same asymptotic limit. 
Our formulation of the method will subsume such considerations.

We also use \Neci's method to derive lower bounds on models with limited
nondeterminism, including branching programs and formulas with limited
nondeterminism as well as to prove limitations on the application of \Neci's
method.
In fact, we show that the Indirect Storage Access function with suitable
parameters yields asymptotically best-possible lower bounds in each of these
models.
To do this, though, we first need to define precisely \emph{what} the \Neci
method actually is, independently from any specific complexity measure, so as to
provide a unifying framework in which it makes sense to speak about \emph{the}
\Neci method. Indeed, although \Neci published his original result 50 years ago
and his ``technique'' has been treated in several classical references
(see~\cite{sav76,we87,we00,ju12}), to the best of our knowledge, there did
not exist any abstract, measure-independent, unifying definition of the
``method'' that would encompass all previous applications of the ``method'' and
allow new ones to be carried out easily, or at least in a clear way. The
definition we suggest is in fact an abstract version of the general definition
that was considered by Alon and Zwick in \cite{alzw89} for the case of
deterministic branching programs.
In this abstract framework we can then show, in a generic way, that for any
complexity measure, an upper bound on the complexity, for this measure, of
computing the Indirect Storage Access function with suitable parameters yields
an upper bound on the best lower bound obtainable using \Neci's method (defined
that way). We then deduce some well-known lower bounds and limitations results,
as well as new ones, in this framework.

\paragraph{Summary of results}
\begin{itemize}
    \item
	We prove that \Neci's method as usually interpreted in the literature
	yields no better nondeterministic or parity branching program size lower bounds
	than $\Theta\Bigr(\frac{n^{3/2}}{\log_2 n}\Bigl)$.
    \item
	We provide a formulation of \Neci's lower bound method upstream from any
	specific complexity measure and link the limitations of this method (in
	terms of the best lower bound obtainable) to upper bounds on the
	complexity for computing one specific family of functions (the Indirect
	Storage Access functions).
    \item
	We apply the method to two classical concrete complexity measures and
	some variants: the size of (deterministic) branching programs (BPs) and
	their nondeterministic (NBPs) and limited nondeterministic (LNBPs)
	and parity ($\oplus$BPs) variants, as well as the size of binary formulas (BFs) and their limited
	nondeterministic variant (LNBFs). See Subsection~\ref{sec:all_models}
	for the formal definitions of these models.
	\begin{table}[h]
	\begin{center}
	\begin{tabular}{|l|c|}
	\hline
	\textbf{Complexity measure} &
	\textbf{Best lower bound obtainable}\\
	\hline
	Size of NBPs & 
	$\Theta\Bigl(\frac{n^{3 / 2}}{\log_2 n}\Bigr)$\\
	\hline
	Size of $\oplus$BPs & 
	$\Theta\Bigl(\frac{n^{3 / 2}}{\log_2 n}\Bigr)$\\
	\hline
	Size of LNBPs using $\Delta(n)$ nondeterministic bits &
	$\Theta\Bigl(\frac{n^2}
	 {2^{\Delta(n)} (\log_2(n) - \Delta(n)) \log_2 n }\Bigr) \,\,
	 (\star)$\\
	\hline
	Size of BPs & 
	$\Theta\Bigl(\frac{n^2}{\log^2_2 n}\Bigr)$\\
	\hline
	Size of LNBFs using $\Delta(n)$ nondeterministic bits &
	$\Theta\Bigl(\frac{n^2}{2^{\Delta(n)} \log_2 n}\Bigr) \,\,
	 (\star)$\\
	\hline
	Size of BFs & 
	$\Theta\bigl(\frac{n^2}{\log_2 n}\bigr)$\\
	\hline
	\end{tabular}
	\end{center}
	\caption{Bounds for the Indirect Storage Access function, which this
		 paper shows to be the best lower bounds obtainable by \Neci's
		 method for any function.
		 The star indicates that the true function is more complicated;
		 however, this current formulation holds for all $\Delta$ not
		 too big with respect to $n \mapsto n$.}
	\label{tbl:Summary_lower_bounds}
	\end{table}
\end{itemize}


\section{Preliminaries}
\label{sec:preliminaries}
For $n \in \N$, we denote by $[n]$ the set of integers from $1$ to $n$, using
the convention that $[0] = \emptyset$.  For $k \in \N_{>0}$ and $\tuple{a} \in
\set{0, 1}^k$ we denote by $\bin_k(\tuple{a})$ its associated natural number
with big-endian representation, i.e. $\sum_{i = 1}^k a_i 2^{k - i}$.
Throughout the paper, the binary representation of a
natural number will refer to its big-endian representation.

Let $V$ be a subset of $\N$. We view $\tuple{a}\in\set{0,1}^V$ as a tuple of
bits of length $\card{V}$ and for $i\in V$ we denote by $a_i$ the corresponding
bit of $\tuple{a}$.

\subsection{Boolean functions and subfunctions}

For $n\in\N$, an $n$-ary \emph{Boolean function over $V$} is a function $f\colon
\set{0, 1}^V \to \set{0, 1}$, where $V\subseteq \N$ and $\card{V}=n$.
When $V$ is not explicit nor implicit we assume that $V=[n]$.
A \emph{family of Boolean functions} is an indexed family $F = \set{f_i}_{i \in
  I}$ where $I \subseteq \N$ and such that for all $i \in I$, $f_i$ is a
Boolean function of arity~$i$.

Subfunctions will play a key role in the lower bound method studied in this
paper, the intuition being that the more different subfunctions a given Boolean
function has, the more difficult it is to compute it.

Let $f$ be an $n$-ary Boolean function.
For any $V \subseteq [n]$ and any $\rho \in \set{0, 1}^{[n] \setminus V}$, we
will denote by $f|_\rho$ the \emph{subfunction of $f$ on $V$ induced by the
partial assignment $\rho$ on $[n] \setminus V$}, that is the function
$f|_\rho\colon \set{0, 1}^{{V}} \to \set{0, 1}$ such that for all
$\tuple{y} \in \set{0, 1}^{{V}}$, we have
$f|_\rho(\tuple{y}) = f(\tuple{x})$ where $x_i = y_i$
for all $i \in V$ and $x_i = \rho(i)$ for all $i \in [n] \setminus V$.
We will also denote by $r_V(f)$ the \emph{total number of subfunctions of $f$
on $V$}, i.e. the cardinality of the set
$s_V(f)=\{f|_\rho \mid \rho \in \{0, 1\}^{[n] \setminus V}\}$.

The following easy lemma gives an upper bound on the total number of
subfunctions of a given Boolean function on a certain subset of input variable
indices.

\begin{lemma}
\label{lem:Number_of_subfunctions_upper_bound}
Let $f\colon \set{0, 1}^n \to \set{0, 1}$ be a Boolean function.
For any $V \subseteq~[n]$,
$r_V(f) \leq \min\set{2^{2^{\card{V}}}, 2^{n - \card{V}}}$.
\end{lemma}
\begin{proof}
Since $r_V(f)$ counts the total number of subfunctions
$f|_\rho\colon \set{0, 1}^{{V}} \to \set{0, 1}$ of $f$ on $V$ induced by a
partial assignment $\rho \in \set{0, 1}^{[n] \setminus V}$, it is at most the
total number of Boolean functions on $\card{V}$ variables (i.e.
$2^{2^{\card{V}}}$) and the total number of assignments to $n - \card{V}$
variables (i.e. $2^{n - \card{V}}$).
\end{proof}

Let $f\colon \set{0, 1}^V \to \set{0, 1}$ be a Boolean function and $i \in
V$. We say that \emph{$f$ depends on its $i$\textsuperscript{th} variable} if
there exist $\tuple{a}, \tuple{a}' \in \set{0, 1}^V$ that differ only in the
bit corresponding to $i$ such that $f(\tuple{a}) \neq f(\tuple{a}')$
(definition based on \cite{ju12}). In particular, if $f$ does not depend on a
set $W$ of variables and $\tuple{a}, \tuple{a}'\in \set{0, 1}^V$ differ only on
bits whose positions are in $W$, then $f(\tuple{a})=f(\tuple{a}')$. The
following proposition shows that variables on which $f$ does not depend do
not affect its number of subfunctions.

\begin{proposition}
\label{lem:Number_of_subfunctions_with_variable_independence}
Let $f\colon \set{0, 1}^n \to \set{0, 1}$ be a Boolean function.
Let $V \subseteq [n]$ and $W \subseteq V$ such that for all
$i \in W$, $f$ does not depend on $x_i$.
Then for $V'=V\setminus W$, $r_V(f) = r_{V'}(f)$.
\end{proposition}
\begin{proof}
We a give bijection from $s_V(f)$ to $s_{V'}(f)$.
For $g\in s_{V}(f)$, say $g=f|_\rho$ for some $\rho\in \{0,1\}^{[n]\setminus V}$,
define $\psi(g)\in s_{V'}(f)$ by
$\psi(g)=g|_\zeta=f|_{\rho\zeta}$ where $\zeta=0^W$ assigns $0$ to all
elements of $W$.
By assumption, for all $i \in W$, $f$ does not depend on $x_i$ so
for all $\zeta'\in\{0,1\}^W$, $f|_{\rho\zeta'}=f|_{\rho\zeta}$; moreover, it is
also easy to see that $f|_{\rho\zeta} = f|_{\rho'\zeta}$ for any
$\rho' \in \{0, 1\}^{[n] \setminus V}$ such that $g = f|_{\rho'}$.
Now let $h'\in s_{V'}(f)$. By definition, there is some
$\rho'\in \{0,1\}^{[n]\setminus V}$ and $\zeta'\in \{0,1\}^W$ such that
$h'=f|_{\rho'\zeta'}$ and by assumption, the latter equals
$f|_{\rho'\zeta}=\psi(g')$ for $g'=f|_{\rho'}$ and hence
$\psi$ is surjective.  

Similarly, for $g, g' \in s_V(f)$ such that $g \neq g'$, we have that there
exists $\tuple{a} \in \{0, 1\}^V$ verifying $g(\tuple{a}) \neq g'(\tuple{a})$.
Let $\rho, \rho' \in \{0, 1\}^V$ such that $g = f|_\rho$ and $g' = f|_{\rho'}$,
and let $\zeta' \in \{0, 1\}^W$ such that $\zeta'(i) = a_i$ for all $i \in W$.
Then
$\psi(g) = f|_{\rho\zeta} = f|_{\rho\zeta'} \neq
 f|_{\rho'\zeta'} = f|_{\rho'\zeta} = \psi(g')$,
so $\psi$ is 1-1 and hence $r_V(f)=\card{s_V(f)}=\card{s_{V'}(f)}=r_{V'}(f)$.
\end{proof}

It will often also be useful to enlarge the size of the domain of a
Boolean function by adding additional input variables on which the function
does not depend in order to obtain complete families of Boolean functions even
if we cannot build a specific Boolean function for each possible input
size. 

\begin{lemma}
\label{lem:Boolean_function_enlargement}
Let $f\colon \set{0, 1}^n \to \set{0, 1}$ be a Boolean function.
Let $n' \in \N$ such that $n' > n$ and let
$f'\colon \set{0, 1}^{n'} \to \set{0, 1}$ be the Boolean function defined by
$f'(\tuple{a}) = f(a_1, \ldots, a_n)$ for all $\tuple{a} \in \{0, 1\}^{n'}$.
Then, for any $V \subseteq [n]$, $r_V(f') = r_V(f)$.
\end{lemma}
\begin{proof}
It is immediate by definition that for $V\subseteq [n]$, $s_V(f')=s_V(f)$.
\end{proof}

\subsection{Hard functions: Indirect Storage Access functions and Element Distinctness}

In this section we define two natural families of functions for which
\Neci's method is known to produce asymptotically optimal lower bounds for
some complexity measures.
The first is the Element Distinctness function.

\begin{definition}
The \emph{Element Distinctness} function $\ED_{N,m}$ for $m \ge N$
is the Boolean function that takes as input $n=N\cdot \lceil \log_2 m\rceil$
bits representing $N$ integers in $[m]$
(outputting $0$ on illegal inputs)
and outputs 1 iff all the $N$ integers
have distinct values.
When $n=2k\cdot 2^k$, we write $\ED_n$ for the function
$\ED_{N,N^2}$ where $N=2^k$.
\end{definition}

The second is the family of Indirect Storage Access functions.
These will turn out to be useful in a broader range of applications
than the Element Distinctness function
and we will see that we can characterize \Neci's method in terms of the bounds
it achieves for these functions.

Indirect Storage Access functions seem to have been originally defined by Paul
in~\cite{pau78} to give an example of a family of Boolean functions for which we
have a trade-off between the minimum sizes of Boolean binary formulas computing
them and the minimum sizes of Boolean binary circuits computing them.  

\begin{definition}
The \emph{Indirect Storage Access} function for $k, \ell \in \N_{>0}$,
denoted $$\isakl\colon \{0, 1\}^{k + 2^k \ell + 2^\ell} \to \{0, 1\}$$ is such that for all $\tuple{a}
\in \{0, 1\}^{k + 2^k \ell + 2^\ell}$, $\isakl(\tuple{a}) = a_{\gamma(\tuple{a})}$ where
$\gamma$ is computed from $\tuple{a}$ as follows:

Let $\alpha(\tuple{a})$ be the number represented in binary by the first $k$
bits of $\tuple{a}$. Let $\beta(\tuple{a})$ be the number represented in binary
by the sequence of $\ell$ bits of $\tuple{a}$ starting at position
$k+ 1 + \ell \alpha(\tuple{a})$. Then $\gamma(\tuple{a})$ is the bit of
$\tuple{a}$ at position $k+ 1 + \ell 2^k + \beta(\tuple{a})$.
Informally speaking, $\isakl$ is just a function reading a bit using two levels
of addressing: a $k$-bit pointer selects an $\ell$-bit pointer (among $2^k$ such
pointers) that picks one bit from a $2^\ell$-bit data string.
\end{definition}

It is known that both these families of Boolean functions yield the
asymptotically strongest lower
bounds obtainable using \Neci's method
for Boolean formula size over arbitrary binary bases, and
deterministic branching program complexity~\cite{we87,bosi90,we00}.
The essence of the argument in each case is the existence of a good partition
with a large count of the number of subfunctions on the variables of the 
partition.   

\begin{lemma}
\label{ED-subfunctions-lemma}
Let $n=2k\cdot 2^k>0$ for $k\in \N$. There is a partition of $[n]$ into blocks
$V_1,\ldots, V_N$ for $N=2^k$  such that for all $i\in [N]$, $\card{V_i}=2k$ and
$r_{V_i}(\ED_n)=\binom{N^2}{N-1}+1\ge N^{N-1}=2^{k(2^k-1)}$.
\end{lemma}

\begin{proof}
Each block in the partition $V_1,\ldots, V_N$ corresponds to the bits of one of
the $N$ numbers for the $\ED_{N,N^2}$ problem.  Observe that for each
assignment of distinct values to the $N-1$ other blocks, the subfunction
induced on the $i$-th block must be different, since precisely those
$N-1$ values must be avoided for the function to have value 1.  There are
$\binom{N^2}{N-1}$ possible choices of those $N-1$ distinct values; for other
assignments, we get the constant $0$ function.
\end{proof}

We now see that for different choices of $k$ and $\ell$, the function
family $\isakl$ provides similar bounds but a more flexible range of parameters
to obtain partitions of different sizes.

\begin{definition}
In the definition of $\isakl$, we will refer to $\alpha(\tuple{a})$ and $\beta(\tuple{a})$ as to the
\emph{primary} and \emph{secondary} pointers of the \isakl\ instance
$\tuple{a}$. The bits of the secondary pointer will be denoted
$\sec_1,\ldots,\sec_\ell$, and more generally the bits of the $p$-th
secondary pointer among the $2^k$ such pointers in the instance at hand
will be denoted $\sec[p]_1,\ldots,\sec[p]_\ell$ for $p\in [2^k]$.
The $2^\ell$ data bits will be referred to as
\emph{Data} and \emph{Data}[$b_1,\ldots,b_\ell$] will stand
for the data bit at position $\bin_\ell(b_1,\ldots,b_\ell) + 1$.
When the context is clear, bits of $\tuple{a}$ will also be viewed as
input variable indices.
\end{definition}

We now see that we can partition the set of input variables of $\isakl$ in such
a way that the number of induced subfunctions is identical and maximal for all
elements of the partition but one: this is formalized in the following lemma.

\begin{lemma}
\label{lem:ISA_function_partitioning}
For every $k, \ell \in \N_{>0}$, there exists a partition $V_1, \ldots, V_{2^k}, U$
of $[k + 2^k \ell + 2^\ell]$ such that $\card{V_i} = \ell$ and
$r_{V_i}(\isakl) = 2^{2^\ell}$ for all $i \in [2^k]$.
\end{lemma}

\begin{proof}
Let $k, \ell \in \N_{>0}$. Consider the partition
$[k + 2^k \ell + 2^\ell]=V_1 \uplus \cdots \uplus V_{2^k}\uplus U$
where $V_i$
is the set $\{\sec[i]_1,\ldots,\sec[i]_\ell\}$ of indices of the $\ell$ variables
forming the $i$\textsuperscript{th} secondary
pointer in the \isakl\ instance $\tuple{a}$.
Then for each setting of the first $k$ variables $a_1,\ldots,a_k$ of
$\tuple{a}$, i.e, for each value $i=\bin_k(a_1,\ldots,a_k)$ of the
primary pointer, every
possible fixing of the $2^\ell$-bit data string induces a different
subfunction on $V_{i + 1}$, hence $r_{V_{i + 1}}(\isakl) = 2^{2^\ell}$.
\end{proof}

From $\isakl$ we define the Indirect Storage Access functions family $\ISA =
\{\ISA_n\}_{n \in \N}$, such that for all $n \in \N$
\begin{itemize}
\item if $n < 5$, $\ISA_n(\tuple{a}) = 0$ for all $\tuple{a} \in \{0, 1\}^n$ ;
\item if there exists $k \in \N_{>0}$ such that $n = h_{\ISA}(k)$, then
      $\ISA_n = \ISA_{k, k + \ceiling{\log_2 k}}$;
\item otherwise,
      $\ISA_n(\tuple{a}) =
       \ISA_{k', k' + \ceiling{\log_2 k'}}(a_1, \ldots, a_{n'})$
      for all $\tuple{a} \in \{0, 1\}^{n}$ where
      $k' = \max\{k \in \N_{>0} \mid h_{\ISA}(k) < n\}$ and $n' = h_{\ISA}(k')$.
\end{itemize}
where 
\[
\function{h_{\ISA}}
	 {\N_{>0}}{\N_{>0}}
	 {m}{m + 2^m (m + \ceiling{\log_2 m}) + 2^{m + \ceiling{\log_2 m}}
	     \displaypunct{.}}
\]

$\ISA$ will be used to give, for each complexity measure we study in this paper
(this notion will be precisely defined in the next subsection), an actual family
of Boolean functions that achieves the best lower bound obtainable using \Neci's
lower bound method (to be defined later). The setting of $k$ and $\ell$ in its
definition is crucial, because if we would for example set $\isan = \isakk$ for
all $n \in \N$ such that there exists $k \in \N_{>0}$ verifying
$n = k + 2^k k + 2^k$, we would not reach the desired bounds.

The next lemma is a simple useful adaptation of
Lemma~\ref{lem:ISA_function_partitioning}.

\begin{lemma}
\label{lem:ISA_partitioning}
For all $n \in \N, n \geq 5$, there exist $p, q \in \N_{>0}$ verifying
$p \geq \frac{1}{32} \cdot \frac{n}{\log_2 n}$ and $q \geq \frac{n}{16}$ such
that there exists a partition $V_1, \ldots, V_p, U$ of $[n]$ such that
$r_{V_i}(\ISA_n) = 2^q$ for all $i \in [p]$.
\end{lemma}
\begin{proof}
  Let $n \in \N, n \geq 5$. Let $k \in \N_{>0}$ be the unique positive integer
  verifying $h_{\ISA}(k) \leq n < h_{\ISA}(k + 1)$ . Set $n'=h_{\ISA}(k)$. By
  definition we have $\ISA_{n'}=\ISA_{k,k+\ceiling{\log_2 k}}$.  Let $V_1,
  \ldots, V_{2^k}, U$ be a partition of $[n']$ such that $r_{V_i}(\ISA_{n'}) = 2^{2^{k
      + \ceiling{\log_2 k}}}$ for all $i \in [2^k]$ as given by Lemma
  \ref{lem:ISA_function_partitioning}.  Moreover, by definition of $\ISA_n$ and
  by Lemma \ref{lem:Boolean_function_enlargement} (for the case in which $n' <
  n$), we have that $r_{V_i}(\ISA_n) = r_{V_i}(\ISA_{n'}) = 2^{2^{k +
  \ceiling{\log_2 k}}}$ for all $i \in [2^k]$.

  Set $q=2^{k + \ceiling{\log_2 k}}$ and $p=2^k$.

  A bit of elementary algebra shows:
 \begin{align*}
n & \leq h_{\ISA}(k+1)
    \leq 16 \cdot 2^{k + \ceiling{\log_2 k}}
   = 16 q
   \leq 16 \cdot 2^{k+\log_2 k + 1}
   \leq 32 k p
\end{align*}
Hence $q \geq \frac{n}{16}$ and $p \geq  \frac{1}{32} \cdot \frac{n}{k} \geq
 \frac{1}{32} \cdot \frac{n}{\log_2 n}$ (as $\log_2 n \geq \log_2 2^k = k$).
\end{proof}

\subsection{Computational models}
\label{sec:all_models}

In this subsection we define the three concrete models of computation
considered extensively in this paper. But first, in view of defining a
model-independent notion of \Neci's method, we define
a complexity measure merely as
a function that associates a non-negative integer
to each Boolean function, as follows.

\begin{definition}
A \emph{complexity measure on Boolean functions} is a function
\begin{equation}
\M\colon \bigcup_{n \in \N} \set{0, 1}^{\set{0, 1}^n} \to \N.
\end{equation}
\end{definition}

Note that the models of computation we consider here are \emph{non-uniform} in
the sense that each computing device only processes inputs of a fixed
length. These models are the following:
\begin{itemize}
\item the nondeterministic branching program (NBP),
\item the parity branching program ($\oplus$BP),
\item the $\delta$-limited nondeterministic branching program ($\delta$-LNBP)
  and
\item the $\delta$-limited nondeterministic Boolean formula ($\delta$-LNBF).
\end{itemize}
The nondeterministic branching program is well known to capture
nondeterministic
logspace $\NL$ when restricted to polynomial size~\cite{pu87}; similarly, when
restricted to polynomial size parity branching programs capture $\PL$.
The two other
models are motivated by the well-known observation that unrestricted
nondeterministic Boolean formulas capture $\NP$~(see
\cite{golemu96}) and further by Klauck's analysis of
restricted nondeterministic fomulas~\cite{kl07}.
Both branching program models extend, albeit in different ways, the
deterministic branching program model known to 
capture deterministic logspace~\cite{co66,ma76}.

\begin{definition}
\label{def:BP}
A \emph{(nondeterministic) branching program (NBP)} on $\{0, 1\}^V$, for a set
$V$ of variables, is a tuple $P = (X, s, t_0, t_1, A_0, A_1, var)$ where
\begin{itemize}
\item $X$ is a finite set of vertices (or states);
\item $s \in X$ is the start (or source) vertex;
\item $t_0, t_1 \in X$, $t_0 \neq t_1$ are two distinct sink vertices;
\item $A_0 \subseteq X \setminus \{t_0, t_1\} \times X \setminus \{s\}$
is the
      set of arcs labelled $0$;
\item $A_1 \subseteq X \setminus \{t_0, t_1\} \times X \setminus
  \{s\}$
is the set of arcs labelled $1$;
\item $var\colon X \setminus \{t_0, t_1\} \to V$ labels
  each non-sink vertex.
\end{itemize}
\end{definition}

\begin{definition}
For a nondeterministic branching program $P = (X, s, t_0, t_1, A_0, A_1, var)$
on $\{0, 1\}^V$, each assignment $\tuple{a}\in\{0,1\}^{V}$ defines a set of arcs
$A[\tuple{a}] = \{(u, v) \in A_0 \mid a_{var(u)} = 0\} \cup
		\{(u, v) \in A_1 \mid a_{var(u)} = 1\}$
and thus a graph $P[\tuple{a}] = (X, A[\tuple{a}])$.
$P$ computes a Boolean function $f\colon\{0,1\}^V\to \{0,1\}$ given by
$f(\tuple{a})=1$ if and only if there exists a path (\emph{computation}) in
$P[\tuple{a}]$ from state $s$ to state $t_1$.

A branching program $P$ defined as above can also be interpreted as a
\emph{parity branching program} that computes a Boolean function
$f^\oplus\colon \{0,1\}^V\to \{0,1\}$ where $f^\oplus(\tuple{a})=1$ if and
only if there is an odd number of paths in $P[\tuple{a}]$ from state $s$ to
state $t_1$.
\end{definition}

\begin{definition}
Branching program $P$ is \emph{deterministic} if and only if
$(X,A_0\cup A_1)$ is acyclic and $A_0$ and $A_1$ each contain precisely one
out-arc from each non-sink vertex of $X$.
$P$ is a \emph{$\delta$-limited nondeterministic branching program} for 
$f\colon\{0,1\}^V\to \{0,1\}$ if and only if $P$ is a deterministic branching
program computing a function $f'\colon\{0,1\}^{V'}\to \{0,1\}$ with
$V\subseteq V'$ and $\card{V'\setminus V}=\delta$ such that
$f(\tuple{a})=
 \bigvee_{\tuple{b}\in \{0,1\}^{V'\setminus V}} f'(\tuple{a},\tuple{b})$.
\end{definition}

\begin{definition}
The complexity measure $\M(f)$ is denoted $\NBP(f)$, $\PBP(f)$, $\LNBP_\delta(f)$ and
$\BP(f)$ for nondeterministic, parity, $\delta$-limited nondeterministic and
deterministic branching programs respectively. $\M(f)$ is
defined in each case 
as the minimum, over every BP of the appropriate type computing $f$,
of the number of non-sink states in (a.k.a.\ \emph{the size of}) the BP.
\end{definition}

Modulo cosmetic details, the above are the standard definitions of
deterministic, nondeterministic and parity Boolean BPs~\cite{we00}.
The definition of
$\delta$-limited nondeterministic BPs, which does not appear to have been
studied previously, is inspired by notions of limited nondeterminism for
other models~\cite{golemu96,hs03:limited-advice,kl07}.

\begin{remark}
The limited nondeterminism of the $\delta$-LNBP model is formulated in
a framework of verification of explicitly represented guesses that is typical
for time-bounded nondeterminism.
In contrast, the NBP model only represents nondeterministic guesses implicitly,
which allows them to be used without being stored, as is typical for 
space-bounded computation.
In particular, even if $\delta$ is unbounded (say $\delta=\infty$), the
smallest $\infty$-LNBP could be somewhat larger than that of an equivalent NBP
and vice-versa.
It is not difficult to see that an NBP of size $s$ can be simulated by such an
$\infty$-LNBP of size at most $2 s^2$. Indeed, simulating the $k$-way branch at
a given state in this NBP in an $\infty$-LNBP can be made by accessing
$\ceiling{\log_2 k}$ fresh nondeterministic bits in a decision tree of size at
most $k$; so since each of the $s$ states of our original NBP branches to at
most $s$ different states for each of the possible values $0$ or $1$, we get
that we can simulate it by an $\infty$-LNBP of size at most $2 s^2$.
Conversely, however, it is unclear by how much the size would increase when
simulating an $\infty$-LNBP by an NBP, but it is widely conjectured to grow
exponentially, since one can prove that polynomial size $\infty$-LNBPs capture
(non-uniform) $\NP$, while polynomial size NBPs capture (non-uniform) $\NL$.
\end{remark}

\begin{remark}
Two other models comparable to the NBP are \emph{contact schemes},
and the more general \emph{switching-and-rectifier
(RS) networks} (see~\cite{ra91,ju12}).
The graph of an RS is undirected and edges either have labels that are
literals or are unlabelled, with
the acceptance condition that of the NBP.  (Contact schemes are a special case
of RS networks that do not have unlabelled edges.)
The size of an RS network is the number of its labelled edges.
One can simulate NBPs by RS networks of at most twice the size
-- each
NBP node becomes an RS node with two labelled children which have unlabelled
edges to the corresponding destination nodes in the NBP.
RS networks and even contact schemes may be smaller than NBPs.
Span programs~\cite{kawi93} can be simulated by parity branching programs
of at most twice the size -- their size is also at most polynomial in 
parity branching program size.
\end{remark}

\begin{definition}
[Deterministic and $\delta$-limited nondeterministic formulas, following
\cite{kl07} and \cite{ju12}]
\label{def:Boolean_LNBFs}
A (deterministic) \emph{Boolean binary formula (BF)} $\varphi$ on $\{0,1\}^n$
($n \in \N$) is a binary tree with
\begin{itemize}
\item a single root,
\item every internal node of arity $2$,
\item every internal node labelled by a function
      $g\colon \{0, 1\}^2 \to \{0, 1\}$,
\item every leaf labelled by one of
      $0,1,x_1,\ldots,x_{n},\neg x_1,\ldots,\neg x_{n}$.
\end{itemize}
$\varphi$ computes a Boolean function $f_\varphi$ on $\{0,1\}^n$ in the
natural way by function composition. 
Let $\delta \in \N$.
A \emph{$\delta$-limited nondeterministic binary formula ($\delta$-LNBF)}
on $\{0, 1\}^n$ $\varphi$ is a deterministic binary formula $\varphi'$
on $\{0,1\}^{n+\delta}$.  It computes a Boolean function
$f_\varphi$ such that for $\tuple{a}\in \{0,1\}^n$,
$f_\varphi(\tuple{a})=
 \bigvee_{\tuple{b}\in \{0,1\}^\delta} f_{\varphi'}(\tuple{a},\tuple{b})$.
\end{definition}

\begin{definition}
The complexity measure $\M(f)$ 
for deterministic and $\delta$-limited nondeterministic formulas
is denoted $\BF(f)$ and $\LNBF_\delta(f)$ respectively. In each case $\M(f)$ is
defined
as the minimum, over every formula $\phi$ of the appropriate type
computing $f$,
of the number $|\phi|$ of non-constant leaves in (a.k.a.\
\emph{the size of}) $\phi$.
\end{definition}

\begin{lemma}
\label{lem:Number_of_subfunctions_proof_checker_upper_bound}
Let $\delta \in \N$ and let $g\colon \{0, 1\}^{n + \delta} \to \{0, 1\}$ and
let $f\colon \{0, 1\}^n \to \{0, 1\}$ be given by 
$f(\tuple{a})=\bigvee_{\tuple{b}\in \{0,1\}^\delta} g(\tuple{a},\tuple{b})$.
Then, for all $V \subseteq [n]$, we have $r_V(f) \leq 
B(r_V(g),2^\delta)-1< r_V(g)^{2^\delta}$
where $B(m,r)=\sum_{i=0}^r \binom{m}{i}$ is the volume of the
Hamming ball of radius $r$ in $\{0,1\}^m$.
\end{lemma}

\begin{proof}
Let $V \subseteq [n]$. For $\rho \in \{0, 1\}^{[n] \setminus V}$, by definition,
$f|_\rho=\bigvee_{\tuple{b}\in \{0,1\}^\delta} g|_{\rho\tuple{b}}$.
Since $\rho\tuple{b}$ assigns all variables in $[n+\delta]\setminus V$, each
function $g|_{\rho\tuple{b}}$ is in $s_V(g)$.
Therefore, each $f|_\rho\in s_V(f)$ is
the $\bigvee$ of $2^\delta$ functions in $s_V(g)$ (not necessarily distinct).
Therefore over all choices of $\rho$, the function $f|_\rho$ only depends on
the set of between 1 and $2^\delta$
among these subfunctions of $g$ that are distinct (and not what values
$\tuple{b}$ with which each such subfunction is associated).
Therefore there are at most $B(r_V(g),2^\delta)-1$ possible distinct
subfunctions $f|_\rho$ in $s_V(f)$.
\end{proof}


\section{Nondeterministic Branching Program Lower Bounds via Shannon Bounds}
In this section we describe the simplest form of 
\Neci's technique and its applications
in order to give some intuition about the technique.
Readers may prefer to skip to the generalised abstract definition of \Neci's
method in Section~\ref{sec:neciporuk}.
The simple version here is based on the so-called ``Shannon bounds" for a
complexity measure.
The Shannon function for a complexity measure maps $n$ to the maximum
complexity of any Boolean function over $\{0,1\}^n$ in that measure.
Lower bounds on the Shannon function typically follow by a simple enumeration
of the number of distinct functions of bounded measure.

For all $n, s \in \N$, and $\M$ a complexity measure let us denote by
$\M_{sem}(n, s)$ the number of distinct $n$-ary
Boolean functions of complexity measure at most $s$.
In particular, define $N_{sem}$ to be the function $\M_{sem}$ for NBPs and
$\oplus_{sem}$ be that for $\oplus$BPs.
The next lemma is the core of the simple version of \Neci's technique.

\begin{lemma}
\label{lem:Weak_Neciporuk_NBP_principle}
For any $n \in \N$, for any $n$-ary Boolean function $f$ on $V$ that depends
on all its inputs and any partition $V_1, \ldots, V_p$ of $V$ we have
$$\NBP(f) \geq
 \sum_{i = 1}^p \max\{\card{V_i},\min\{s \in \N \mid N_{sem}(\card{V_i}, s) \geq r_{V_i}(f)\}\}.$$
$$\PBP(f) \geq
 \sum_{i = 1}^p \max\{\card{V_i},\min\{s \in \N \mid \oplus_{sem}(\card{V_i}, s) \geq r_{V_i}(f)\}\}.$$
\end{lemma}

\begin{proof}
Let $n \in \N$, $f$ be an $n$-ary Boolean function on $V$ depending on all
its inputs and $V_1, \ldots, V_p$ a partition of $V$.
Let $P$ be a Boolean NBP computing $f$. For all $i \in [p]$ we will denote by
$s_i \in \N$ the number of vertices in $P$ labelled by elements in $V_i$. It is
clear that $P$ is of size at least $\sum_{i = 1}^p s_i$.

Let $i \in [p]$. Observe that for all subfunction $g$ of $f$ on $V_i$,  there is
by definition a partial assignment $\rho$ on $V \setminus V_i$ such that
$f|_\rho = g$, so it is not too difficult to see that $g$ is computed by the
Boolean NBP of size $s_i$ obtained from $P$ by:
\begin{enumerate}
    \item
	removing all non sink vertices labelled by elements not in $V_i$;
    \item
	defining the new start vertex as one of the vertices whose label is in
        $V_i$ and connected to the start vertex of $P$ by a path of nodes
	labelled by a variable outside of $V_i$ and arcs labelled consistently
	with $\rho$ and adding both an arc labelled 0 and an arc labelled 1 from
	this new start vertex to each other such reachable vertex (except for
	the extreme case of a constant function, in which we just set the start
	vertex as the appropriate sink vertex);
    \item
	connecting a vertex $u$ to a vertex $v$ by an arc labelled by
	$a \in \set{0, 1}$ if and only if there exists a path from $u$ to $v$
	in $P$ verifying that any intermediate vertex of the path is labelled by
	a variable outside of $V_i$, the first arc is labelled by $a$ and each
	arc (but the first one) is labelled consistently with $\rho$.
\end{enumerate}
Thus, $r_{V_i}(f)$ is necessarily upper-bounded by the number of semantically
distinct such NBPs we can build from $P$ that way, which is in turn at most
$N_{sem}(\card{V_i}, s_i)$.
Moreover, since, by construction, $f$ depends on all variables whose indices are
in $V_i$, we have that for each element $\ell \in V_i$, $P$ contains at least
one vertex labelled by $\ell$, so $s_i \geq \card{V_i}$.
Hence, for each $i$,
$s_i \geq \max\{\card{V_i},\min\{s \in \N \mid
		N_{sem}(\card{V_i}, s) \geq r_{V_i}(f)\}\}$.
Since the NBP has size at least $\sum_{i=1}^p s_i$ and the NBP is arbitrary,
the bound of the lemma follows.

The same argument also applies directly to yield the bound for $\oplus$BPs, with
$\oplus_{sem}(\card{V_i},s)$ replacing $N_{sem}(\card{V_i},s)$.
\end{proof}

\begin{proposition}
\label{Shannonbound}
Let $s\ge n$.  Then
$N_{sem}(n, s), \oplus_{sem}(n,s) < 2^{2 (s+1)^2}$.
\end{proposition}

\begin{proof}
We simply count the number of distinct branching programs.
Subject to renaming and reorganising, any $n$-ary Boolean function
computable by an NBP or $\oplus$BP
of size at most $s$, can be
computed by one of size exactly $s$, having
$\{v_j\}_{j = 1}^{s + 2}$ as vertices, $v_1$ as start vertex, $v_{s + 1}$ as
$0$-vertex and $v_{s + 2}$ as $1$-vertex, where no arc goes to the $0$-vertex.
The out-edges at each node $v_i$ can be described by the subset of
vertices
$v_j$, $j\ne i$, reached on each of values $0$ and $1$.
There are
$(s-1)!$ different ways of reordering the names of vertices
$v_2,\ldots, v_s$ that keep identical connectivity of the branching program
and hence the function it computes, both as an NBP and a $\oplus$BP.
Hence, it directly follows that
$N_{sem}(n, s),\oplus_{sem}(n,s) \leq
 (2^{2s} n)^{s}/(s-1)!\le 2^{2s^2} s^{s}/(s-1)!$,
since $s \geq n$, therefore, since $s!\ge (s/e)^s$,
$N_{sem}(n, s),\oplus_{sem}(n,s) \leq 2^{2 s^2} s e^s < 2^{2 (s+1)^2}$.
\end{proof}

\begin{definition}
\label{def:Weak_Neciporuk_NBP}
For the complexity measures $\M=\NBP, \PBP$, the \emph{simple \Neci\ lower bound method}
consists of the following.
\begin{enumerate}
    \item
	Giving explicitly a non-decreasing function $b\colon \N_{>0} \to \N$
	such that for any $n \in \N$, for any $n$-ary Boolean function $f$ on
	$V$ that depends on all its inputs and any partition
	$V_1, \ldots, V_p$ of $V$, we have
	$\sum_{i = 1}^p
	 \max\{\card{V_i},\min\{s \in \N \mid
	       \M_{sem}(\card{V_i}, s) \geq r_{V_i}(f)\}\} \geq
	 \sum_{i = 1}^p b(r_{V_i}(f))$.
    \item
	For a given $n$-ary Boolean function $g$ on $V$ that depends on all
	the variables whose indices are in $V$, explicitly choosing a
	partition $V_1, \ldots, V_p$ of $V$,
	computing $r_{V_i}(g)$ for all $i \in [p]$ and
	concluding that $\M(g) \geq \sum_{i = 1}^p b(r_{V_i}(f))$.
\end{enumerate}
\end{definition}

A function $b$ satisfying the condition of Step 1 in the definition above is
called a \emph{simple \Neci\ function for $\M$}.

We now give an explicit simple \Neci\ function for $\NBP$ and $\PBP$.

\begin{proposition}
\label{ptn:Weak_Neciporuk_NBP_size_lower_bound}
The function on $\N_{>0} \to \N$ given by
$m \mapsto \ceiling{\sqrt{\frac{1}{2} \log_2 m}-1}$ is a simple \Neci\ function for
$\NBP$ and for $\PBP$.
\end{proposition}

\begin{proof}
We start by observing that the function on $\N_{>0} \to \N$ given by
$m \mapsto \ceiling{\sqrt{\frac{1}{2} \log_2 m}-1}$ is obviously non-decreasing.
Let $n \in \N$, $f$ be an $n$-ary Boolean function $f$ on $V$ depending on
all its variables and $V_1, \ldots, V_p$ be a partition of $V$.
Let $i \in [p]$.
Let
$s_i = \max\{\card{V_i},\min\{s \in \N \mid N_{sem}(\card{V_i}, s) \geq r_{V_i}(f)\}\}$
for all $i \in [p]$. We claim that
$s_i \geq \ceiling{\sqrt{\frac{1}{2} \log_2 r_{V_i}(f)}-1}$ for all $i \in [p]$.

By definition, $N_{sem}(\card{V_i}, s_i) \geq r_{V_i}(f)$, and since
$s_i\ge \card{V_i}$, Proposition~\ref{Shannonbound} implies that
$N_{sem}(\card{V_i}, s_i) < 2^{2 (s_i+1)^2}$ and hence
$2^{2 (s_i+1)^2} > r_{V_i}(f)$, that is to say,
$s_i > \sqrt{\frac{1}{2} \log_2 r_{V_i}(f)}-1$.
Since $s_i$ is integral, we deduce that
$s_i \geq \ceiling{\sqrt{\frac{1}{2} \log_2 r_{V_i}(f)}-1}$.
The lemma follows for $\NBP$; the argument for $\PBP$ is identical replacing
$N_{sem}$ by $\oplus_{sem}$.
\end{proof}

This directly gives us the following lower bounds.

\begin{proposition}[\cite{pu87,kawi93}]

$\NBP(\ED_n),\NBP(\ISA_n),\PBP(\ED_n),\PBP(\ISA_n) \in
 \Omega\Bigl(\frac{n^{3/2}}{\log_2 n}\Bigr)$.
\end{proposition}

\begin{proof}
We first consider $\ED_n$ for $n=2k2^k$ and $k\ge 2$. 
By Lemma~\ref{ED-subfunctions-lemma} there is a partition $V_1,\ldots, V_N$
of $[n]$ such that $r_{V_i}(\ED_n)\ge N^{N-1}$ and $\card{V_i}=2k$ for all
$i \in [N]$. 
Applying Proposition~\ref{ptn:Weak_Neciporuk_NBP_size_lower_bound}, since
$\ED_n$ depends on all its variables, we have
\begin{align*}
\NBP(\ED_n), \PBP(\ED_n) &\geq \sum_{i=1}^N \ceiling{\sqrt{\frac{1}{2} \log_2 r_{V_i}(\ED_n)}-1}
 \geq N \cdot \ceiling{\sqrt{\frac{1}{2} \log_2 N^{N-1}}}-N\\
&\geq N \cdot \sqrt{\frac{N-1}{2} \log_2 N}-N
\end{align*}
which is in $\Omega(N^{3/2}(\log_2 N)^{1/2})$ and hence
$\Omega\Bigl(\frac{n^{3/2}}{\log_2 n}\Bigr)$ since $n$ is
$O(N \log_2 N)$.

We now consider $\ISA$.
Let $n \in \N, n \geq 32$.
Let $V_1, \ldots, V_p, U$ be a partition of $[n]$ such that
$r_{V_i}(\isan) = 2^q$ for all $i \in [p]$ where $p, q \in \N_{>0}$ verify
$p \geq \frac{1}{32} \cdot \frac{n}{\log_2 n}$ and $q \geq \frac{n}{16}$ as
given by Lemma \ref{lem:ISA_partitioning}.
Applying Proposition~\ref{ptn:Weak_Neciporuk_NBP_size_lower_bound}, since
$\ISA_n$ depends on all its variables, we have
\begin{align*}
\NBP(\ISA_n), \PBP(\ISA_n)
& \geq \sum_{i = 1}^p \ceiling{\sqrt{\frac{1}{2} \log_2 r_{V_i}(\isan)}-1} \\
& \geq \sum_{i = 1}^p \Bigl(\sqrt{\frac{1}{2} \log_2(2^q)}-1\Bigr)
 = p\cdot \Bigl(\sqrt{\frac{q}{2}} -1\Bigr)\\
& \geq \frac{1}{32} \cdot \frac{n}{\log_2 n} \cdot
       \Bigl(\sqrt{\frac{n}{32}}-1\Bigr)
\displaypunct{.}
\end{align*}
So $\NBP(\ISA_n), \PBP(ISA_n) \in \Omega\Bigl(\frac{n^{3/2}}{\log_2 n}\Bigr)$.
\end{proof}

To understand the best lower bounds we can prove with the simple \Neci\
lower bound method, we first give a lower bound
on $N_{sem}(n, s)$ and $\oplus_{sem}(n,s)$ (valid for suitable values of $n, s \in \N$) that will allow
us to give an upper bound on all simple \Neci\ functions for $\NBP$, $\PBP$.
We do this by
giving an easy upper bound on the size needed by NBPs and $\oplus$BPs
to compute any $n$-ary Boolean function\footnote{Note that there are somewhat
tighter but more complicated upper bounds of $2^{n/2+1}$ for $\NBP$ due to
Lupanov~\cite{lup58} and a tight asymptotic upper bound of $2^{(n+1)/2}$ for $\PBP$
due to \Neci~\cite{ne62}, respectively; see~\cite{ju12}.}; i.e., simple
upper bounds on the Shannon function for $\NBP$ and $\PBP$.

\begin{lemma}
\label{lem:NBP_overall_upper_bound}
For any $n$-ary Boolean function $f$ on $\{0,1\}^n$ ($n \in \N$),
$\NBP(f),\PBP(f) \leq 3\cdot 2^{\ceiling{\frac{n}{2}}}$.
\end{lemma}

\begin{proof}
Assume that $n=2t$ is even.
The constructed NBPs will have only one nondeterministic level, will be
the same for all functions for the other
levels 1 to $t-1$ and $t+1$ to $2t$, and every vertex at each level $i$ will
query variable $x_i$.

The first $t-1$ levels form a complete decision tree of height $t-1$ on
variables $x_1,\ldots, x_{t-1}$ with a vertex at level $t$ for each assignment
$a_1\ldots a_{t-1}$ to these variables.
The last $t$ levels of the NBP consist of a complete fan-in tree of height t
on variables $x_{t+1},\ldots,x_{2t}$ as follows: 
there is a vertex at level $t'>t$ for every assignment $a_{t'}\ldots a_{2t}$
to $x_{t'}, \ldots, x_t$ and there is an out-arc labelled $a_{t'}$ from
this vertex to the vertex at level $t'+1$ corresponding to
$a_{t'+1}\ldots a_{2t}$.
The $1$-output vertex has two in-arcs, one labelled $a_{2t}$
from each vertex $a_{2t}$ at level $2t$.

Finally, we define the nondeterministic level $t$ of the NBP for function $f$.
For each assignment $a_1\ldots a_{2t}$ on which $f$ evaluates to $1$, there
is an out-arc labelled $a_t$ from the vertex corresponding to
$a_1\ldots a_{t-1}$ at level $t$ (which queries $x_t$) to the vertex
corresponding to $a_{t+1}\ldots a_{2t}$ at level $t+1$. 

The constructed NBP has at most $3\cdot 2^t=3\cdot 2^{\frac{n}{2}}$ vertices.
By observing that there is precisely one accepting path on any accepted
input, we see that it is also a $\oplus$BP.
\end{proof}

\begin{corollary}
\label{cor:Nsem_lower_bound}
For all $n, s \in \N$, $n\ge 2\floor{\log_2(\frac{s}{3})}$,
$\Nsem(n,s),\oplus_{sem}(n,s)> 2^{\frac{s^2}{36}}$.
\end{corollary}

\begin{proof}
Clearly $\Nsem(n,s)$ is non-decreasing in $n$,
so it suffices to prove the corollary for $n = 2 \floor{\log_2(\frac{s}{3})}$.
Then $3\cdot 2^{\frac{n}{2}}\le s < 6\cdot 2^{\frac{n}{2}}$.
There are precisely $2^{2^{n}} > 2^{\frac{s^2}{36}}$ different Boolean functions
on $n$ inputs and, by Lemma~\ref{lem:NBP_overall_upper_bound}, each may be
computed by an NBP or $\oplus$BP of size at most $s$.
\end{proof}

\begin{theorem}
\label{thm:Weak_Neciporuk_limitation_NBP}
Let $F = \{f_n\}_{n \in \N}$ be a family of Boolean functions.
Let $L\colon \N \to \N$ be such that for each $n \in \N$, the lower bound $L(n)$
for $\NBP(f_n)$ or $\PBP(f_n)$ has been obtained using the simple \Neci lower
bound method.
Then, $L(n) \in \Omicron\Bigl(\frac{n^{3/2}}{\log_2 n}\Bigr)$.
\end{theorem}

\begin{proof}
Let $F = \{f_n\}_{n \in \N}$ be a family of Boolean functions, where for each
$n \in \N$, $f_n$ depends on all the variables in $D_n\subseteq [n]$.
Let $L\colon \N \to \N$ be such that
\[
L(n) = 
\max\Bigl\{\sum_{i = 1}^p \max\{\card{V_i},\min\{s \in \N \mid
				N_{sem}(\card{V_i}, s) \geq r_{V_i}(f_n)\}\}
	   \mathrel{\Big|}
	   V_1, \ldots, V_p \text{ partition } D_n \Bigr\}
\]
for all $n \in \N$.

Let $n \in \N$ and $V_1, \ldots, V_p$ be a partition of $D_n$.
Let $i \in [p]$ and set
$$s_i = \max\{\card{V_i},
\min\{s \in \N \mid N_{sem}(\card{V_i}, s) \geq r_{V_i}(f_n)\}\}.$$
Suppose that $s_i=
\min\{s \in \N \mid N_{sem}(\card{V_i}, s) \geq r_{V_i}(f_n)\}\}>\card{V_i}$.
We now have two cases depending on $\card{V_i}$:
If $\card{V_i} < \log_2\log_2 r_{V_i}(f) + 3$ then, by
Lemma~\ref{lem:NBP_overall_upper_bound}, since circuits of size
$3 \cdot 2^{\ceiling{\frac{\card{V_i}}{2}}}$ suffice to compute all
functions on $V_i$, which include those counted in $r_{V_i}(f_n)$, 
$$s_i\le 3 \cdot 2^{\ceiling{\frac{\card{V_i}}{2}}}
\le 3 \cdot 2^{(\log_2\log_2 r_{V_i}(f))/2 + 2} =
 12 \sqrt{\log_2 r_{V_i}(f)}.$$
On the other hand, if $\card{V_i} \geq \log_2\log_2 r_{V_i}(f) + 3$ then
setting $s = \ceiling{6 \sqrt{\log_2 r_{V_i}(f)}}$, then
$$2 \log_2(\frac{s}{3}) \leq
 2 \log_2\Bigl(\frac{6 \sqrt{\log_2 r_{V_i}(f)} + 1}{3}\Bigr)
 \leq 2 \log_2\Bigl(\frac{7}{3} \sqrt{\log_2 r_{V_i}(f)}\Bigr) \leq
 \log_2\log_2r_{V_i}(f) + 3 \leq \card{V_i}$$
so, by Corollary~\ref{cor:Nsem_lower_bound}, we have that
$N_{sem}(\card{V_i}, s) > 2^{\frac{s^2}{36}} \geq r_{V_i}(f)$, which means that
$s_i \leq \ceiling{6 \sqrt{\log_2 r_{V_i}(f)}}$.
Therefore, for all $n \in \N$,
\[
L(n) \leq 
\max\Bigl\{\sum_{i = 1}^p
	   \max\{\card{V_i},12 \sqrt{\log_2 r_{V_i}(f)}\}
	   \mathrel{\Big|}
	   V_1, \ldots, V_p \text{ partition } D_n \Bigr\}
\displaypunct{.}
\]
Let $n \in \N, n \geq 4$.
By Proposition~\ref{lem:Number_of_subfunctions_upper_bound}, it follows that
\begin{align*}
L(n)
\leq & \max\Bigl\{\sum_{i = 1}^p
		  \max\{\card{V_i},12 \sqrt{\log_2(\min\{2^{2^{\card{V_i}}},
						       2^{n - \card{V_i}}\})}\}
		  \mathrel{\Big|}
		  V_1, \ldots, V_p \text{ partition }
		  D_n \Bigr\}\\
= & \max\Bigl\{\sum_{i = 1}^p
	       \max\{v_i,12 \sqrt{\log_2(\min\{2^{2^{v_i}},
						    2^{n - v_i}\})} \}
						    \mathrel{\Big|}
	       \sum_{i = 1}^p v_i \le n  \text{ and }
	       \forall i \in [p], v_i > 0\Bigr\}\\
\leq & \max\Bigl\{\sum_{i = 1}^p
		  \max\{v_i,12 \sqrt{\min\{2^{v_i}, n - v_i\}}\}
		  \mathrel{\Big|}
		  \sum_{i = 1}^p v_i = n \text{ and  }
		  \forall i \in [p], v_i > 0\Bigr\}\\
= & \max\Bigl\{\sum_{i = 1}^p h(v_i) \mathrel{\Big|}
	       \sum_{i = 1}^p v_i = n \text{ and  }
	       \forall i \in [p], v_i > 0\Bigr\}
\tag{$\star\star$} \label{eq:Max_weak_Neciporuk_lower_bound}
\end{align*}
where $h(v) = \max\{v,12 \min\{2^{v/2}, \sqrt{n - v}\}\}$ for
all $v \in \N$.

Let now $v_1, \ldots, v_p$ realise the maximum in
\eqref{eq:Max_weak_Neciporuk_lower_bound}.
Clearly, $h(v) = 12 \cdot 2^{v/2}$ for all
$v \in \N, v \leq \log_2 n - 1$ and hence  $h(v) + h(v') \leq h(v + v')$ if
$v + v' \leq \log_2 n - 1$. It follows that without loss of generality we can
assume that there exists at most one $j \in [p]$ such that $v_j$ is smaller than
$\frac{\log_2 n - 1}{2}$. Such a small $v_j$ has
$h(v_j) =
 12 \cdot 2^{v_j/2} < 12 \cdot n^{1/4}$.
Let now $I \subseteq [p]$ such that $i \in I$ if and only if $h(v_i) = v_i$.
We have that $\sum_{i \in I} h(v_i) \leq n$ by definition of $v_1, \ldots, v_p$.
Moreover, in $[p] \setminus I$, there are at most $\frac{2 n}{\log_2 n - 1}$
elements, since for all $i \in I \setminus \{j\}$,
$v_i \geq \frac{\log_2 n - 1}{2}$, 
and each such $i$ verifies
$h(v_i) \leq 12 \sqrt{n - v_i} \leq 12 \sqrt{n}$.
Hence,
\[
L(n) \leq
\sum_{i = 1}^p h(v_i) \leq
24 \cdot \frac{n^{3/2}}{\log_2 n - 1} +
12 \cdot n^{1/4} + n
\leq 74 \cdot \frac{n^{3/2}}{\log_2 n}
\]
which completes the proof.
\end{proof}

\subsection*{Limitations of this Formulation}

Simply using some adaptation of
Definition~\ref{def:Weak_Neciporuk_NBP} would not allow us to recover the
well-known $\Omega\bigl(\frac{n^2}{\log n}\bigr)$ lower bound for size of
binary formulas contained in \Neci's original article~\cite{ne66}.
Indeed, for all $n, s \in \N$, let us denote by $F_{sem}(n, s)$ the number of
$n$-ary Boolean functions on some fixed $V$ computable by BFs of size at most
$s$. We can prove a Lemma analogous to
Lemma~\ref{lem:Weak_Neciporuk_NBP_principle} where $\NBP$ is replaced by $\BF$
and $N_{sem}$ by $F_{sem}$.
Similarly, we can define the \emph{simple \Neci lower bound method} for $\BF$ as in
Definition~\ref{def:Weak_Neciporuk_NBP}, as well as \emph{simple \Neci functions
for $\BF$} accordingly.
However, Lupanov showed (see~\cite[p.31]{ju12}) that for all
$n \in \N$, any $n$-ary Boolean function on some $V$ can be computed by a BF of
size at most $\alpha \cdot \frac{2^n}{\log_2 n}$ for some constant
$\alpha \in \R_{>0}$ (a result which is analogous to
Lemma~\ref{lem:NBP_overall_upper_bound}). Following the same strategy as for the
proofs of Corollary \ref{cor:Nsem_lower_bound} and Theorem
\ref{thm:Weak_Neciporuk_limitation_NBP}, we can show that this implies there
exists a constant $\beta \in \R_{>0}$ such that any simple \Neci function
$b\colon \N_{>0} \to \N$ for $\BF$ verifies
$b(m) \leq \beta \cdot \frac{\log_2 m}{\log_2\log_2\log_2 m}$
for any sufficiently large $m \in \N_{>0}$.
This means that this does not allow us to
recover the well-known \Neci bounding function of
$m \mapsto \frac{1}{4} \log_2 m$ (see e.g. \cite[Theorem 6.16]{ju12}), and
therefore also not \Neci's original lower bound.

Even if we managed to adapt
Definition~\ref{def:Weak_Neciporuk_NBP} to the case of binary formulas, we
cannot really do it in a clean way for all complexity measures we would like to
study. 
If we were to try to adapt
Lemma~\ref{lem:Weak_Neciporuk_NBP_principle} to the case of 
the size of limited nondeterministic branching programs (LNBPs), we would
define, as usual, for all $n, s, \delta \in \N$, the number
$LN_{sem}(n, s, \delta)$ of $n$-ary Boolean functions on some fixed $V$
computable by LNBPs of size at most $s$ and using $\delta$ nondeterministic
bits. But then, it would be false to say that for any $\delta, n \in \N$, for
any $n$-ary Boolean function $f$ depending on all of $V$ and any partition
$V_1, \ldots, V_p$ of $V$ we have
$\LNBP_\delta(f) \geq
 \sum_{i = 1}^p \max\{\card{V_i},
		      \min\{s \in \N \mid
			    LN_{sem}(\card{V_i}, s, \delta) \geq r_{V_i}(f)\}\}$
(this would induce an overcount, as we would most certainly count vertices
corresponding to nondeterministic variables several times).

These considerations led us to the more general
formulation of the \Neci\ method described in the next section.


\section{An abstract formulation of \Neci's method}
\label{sec:neciporuk}
In this section we present an abstract version of \Neci's lower bound method and
provide some model-independent meta-results on the limitations of this method.

The main idea of the general version of the method is, for a given Boolean
function, to partition
its set of input variables and to lower bound its complexity by a sum over each
element of the partition of a partial cost that depends only on the number of
subfunctions of the function on the variables in this element. More formally, we
state the method in the following way.

\begin{definition}\label{defi-neci}
For a given complexity measure $\M$ on Boolean functions, \emph{\Neci's
lower bound method} consists of the following.
\begin{enumerate}
\item Giving explicitly a non-decreasing function $b\colon \N_{>0} \to \N$ such
  that for any $n \in \N$, for any $n$-ary Boolean function $f$ and any
  partition $V_1, \ldots, V_p$ of~$[n]$, we have $\M(f) \geq \sum_{i =
    1}^p b(r_{V_i}(f))$.
\item
For a given $n$-ary Boolean function $g$, explicitly choosing a partition $V_1, \ldots, V_p$ of $[n]$,
computing $r_{V_i}(g)$ for all $i \in [p]$ and concluding that
$\M(g) \geq \sum_{i = 1}^p b(r_{V_i}(g))$.
\end{enumerate}
\end{definition}

A function $b$ satisfying the condition of Step~1 in \Neci's method is called a
\emph{\Neci function for $\M$} and we denote by $\NeciporukSet_{\M}$ the
\emph{set of all \Neci functions for $\M$}.

The first step of Definition~\ref{defi-neci} is usually not included in the
\Neci method. For instance in~\cite{we87,ju12}, an explicit \Neci function
$b$ is given for a complexity measure $\M$ and therefore the result concerning
the limitation of the method is relative to this function $b$. In the case of
deterministic branching programs, the best possible $b$ was given
by Alon and Zwick~\cite{alzw89}, who use a similar definition but we are not aware of any result of this kind for other
complexity measures.

It follows from Definition~\ref{defi-neci} that the best lower bound achievable by the
\Neci method for a family $F=\{f_n\}_{n\in\N}$ of Boolean functions and a complexity
measure $\M$ is the function $\NeciporukLB^{\M}_F$:
\begin{equation}
n  \mapsto  \max\set{\sum_{i = 1}^p b(r_{V_i}(f_n)) \mid
		    b \in \mathcal{N}_{\M} \wedge
		    V_1, \ldots, V_p \text{ partition of } [n]}
                \end{equation}

\subsection{Meta-results on \Neci's method}\label{sec-neci-meta}

We now give two results concerning \Neci's method depending on
hypotheses on the complexity measure $\M$.  We will apply those results in the
next section with the appropriate constants and functions for each of the
concrete computational models we consider in this paper.

The first meta-result is that an upper bound on the
complexity of the functions $\isakk$ implies an upper bound on
every $b \in 
\NeciporukSet_{\M}$. Intuitively this is possible because by definition,
$b$ entails a lower bound on $\M(f)$ for every function $f$.

\begin{lemma}
\label{lem:Neciporuk's_method_function_limitation_meta-result}
Let $\M$ be a given complexity measure on Boolean functions and assume that we
have a non-decreasing function
$g_\M\colon \cointerval{1}{+\infty} \to \R_{\geq 0}$ such that
$\M(\ISA_{k, k}) \leq g_\M(k)$ for all $k \in \N_{>0}$ and there exists a constant
$\alpha \in \R_{>0}$ such that $\frac{g_\M(k + 1)}{g_\M(k)} \leq \alpha$ for all
$k \in \N_{>0}$.
Then, any $b \in \NeciporukSet_{\M}$ is such that
\[b(m) \leq \alpha \cdot \frac{g_\M(\log_2\log_2 m)}{\log_2 m }\]
for all $m \in \N, m \geq 4$.
\end{lemma}

\begin{proof}
  Let $b \in \NeciporukSet_{\M}$. Let $m \in \N, m \geq 4$ and $k\in \N_{>0}$ be
  such that $2^{2^k} \leq m \leq 2^{2^{k + 1}}$. Hence $2^k \leq \log_2 m \leq
  2^{k + 1}$ and of course $k \leq \log_2\log_2 m \leq k + 1$.  Consider now
  $\ISA_{k + 1, k + 1}$. By Lemma~\ref{lem:ISA_function_partitioning} we have a
  partition $V_1, \ldots, V_{2^{k + 1}}, U$ of the set of indices $[(k + 1) +
  2^{k + 1} (k + 1) + 2^{k + 1}]$ of the input variables of $\ISA_{k + 1, k +
    1}$ such that $r_{V_i}(\ISA_{k + 1, k + 1}) = 2^{2^{k + 1}}$ for all $i \in
  [2^{k + 1}]$. By hypothesis, it therefore follows that:
\begin{align*}
g_\M(k + 1)
& \geq \M(\ISA_{k + 1, k + 1})\\
& \geq \sum_{i = 1}^{2^{k + 1}} b(r_{V_i}(\ISA_{k + 1, k + 1})) +
       b(r_U(\ISA_{k + 1, k +1}))\\
& \geq \sum_{i = 1}^{2^{k + 1}} b\bigl(2^{2^{k + 1}}\bigr)\\
& = 2^{k + 1} b\bigl(2^{2^{k + 1}}\bigr)
\displaypunct{,}
\end{align*}
therefore $b(2^{2^{k + 1}}) \leq \frac{g_\M(k + 1)}{2^{k + 1}}$.
But since $\frac{g_\M(k + 1)}{g_\M(k)} \leq \alpha$,
$g_\M(k) \leq g_\M(\log_2\log_2 m)$ (because $g_\M$ is non-decreasing and
$1 \leq k \leq \log_2\log_2 m$),  $b$ is non-decreasing and $m \leq 2^{2^{k +
    1}}$ we have:
\[
b(m) \leq b\bigl(2^{2^{k + 1}}\bigr) \leq \frac{g_\M(k + 1)}{2^{k + 1}} \leq
\alpha \cdot \frac{g_\M(\log_2\log_2 m)}{\log_2 m}
\displaypunct{.}
\]
In conclusion, we indeed have that
$b(m) \leq \alpha \cdot \frac{g_\M(\log_2\log_2 m )}{\log_2 m}$
for all $m \in \N, m \geq 4$.
\end{proof}

Assuming an upper bound on every $b \in \NeciporukSet_{\M}$, as given for
instance by the previous lemma, we can derive an upper bound on
$\NeciporukLB^{\M}_F$ independently of the family of Boolean functions
$F$. That is to say that we can give an overall (asymptotic) upper bound on the
best complexity lower bounds we may obtain using \Neci's lower bound method for
the complexity measure $\M$, exhibiting the limitation of the method.

\begin{lemma}
\label{lem:Neciporuk's_method_limitation_complexity_lower_bound_meta-result}
Let $\M$ be a given complexity measure on Boolean functions and assume that we
have a function $h_\M\colon \cointerval{4}{+\infty} \to \R_{\geq0}$ such that there
exist $x_0 \in \cointerval{2^8}{+\infty}$ and a constant $\alpha \in \R_{>0}$
verifying that:
\begin{enumerate}[(i)]
\item $h_\M$ is non-decreasing on $\cointerval{x_0}{+\infty}$;
\item $h_\M(2^x) \geq \log_2 x$ for all $x \in \cointerval{\log_2 x_0}{+\infty}$;
\item $h_\M(2^{2^v}) + h_\M(2^{2^{v'}}) \leq h_\M(2^{2^{v + v'}})$ for all $v, v' \in \N$
      verifying $2^{2^v} \geq x_0$ and $2^{2^{v'}} \geq x_0$;
\item for all $b \in \NeciporukSet_{\M}$ and $m \in \N, m \geq 4$, we have $b(m) \leq \alpha \cdot h_\M(m)$.
\end{enumerate}
    Then, for any family of Boolean functions $F = \{f_n\}_{n \in \N}$, we have
\[
\NeciporukLB^{\M}_F(n) \leq
\alpha \cdot \Bigl(4 + h_\M(\floor{x_0})\Bigr) \cdot
\frac{n}{\log_2 n} \cdot h_\M(2^n)
\]
for all $n \in \N, n \geq \log_2 x_0$.
\end{lemma}

\begin{proof}
The condition $n \geq \log_2 x_0$ ensures that $h_\M(2^n)$ is always well defined
and satisfies~(ii).
Let $F = \{f_n\}_{n \in \N}$ be a family of Boolean functions. For all
$n \in \N_{>0}$, let $h'_n\colon [n] \to \R$ be the function defined on $[n]$ by
\begin{equation*}
h'_n(v) = \begin{cases}
	  \alpha \cdot \min\{h_\M(2^{2^v}), h_\M(2^{n - v})\} &
	  \text{if $\min\{2^{2^v}, 2^{n - v}\} \geq x_0$}\\
	  \alpha \cdot h_\M(\floor{x_0}) & \text{otherwise.}
	  \end{cases}
\end{equation*}

\begin{myclaim}\label{claim-h-prime}
If $n \in \N_{>0}$ and $v \in [n]$ are such that
$v\leq \log_2 n - 1$ and $2^{2^v}\geq x_0$ then
$h'_n(v) = \alpha \cdot h_\M(2^{2^v})$.
\end{myclaim}
\begin{proof}
With the hypothesis of the claim we have:
\[
2^{2^v} \leq 2^{2^{\log_2 n - 1}} = 2^{\frac{n}{2}} \leq
2^{n - \log_2 n + 1} \leq 2^{n - v}
\displaypunct{,}
\]
the middle inequality being a consequence of $n \geq  
\log_2 x_0 \geq 8$. Hence in this case
$\min\set{2^{2^v}, 2^{n - v}} \allowbreak = 2^{2^v}$
which is greater  than  $x_0$. As by (i) $h_\M$ is non-decreasing we
have $h'_n(v) = \alpha \cdot h_\M(2^{2^v})$.
\end{proof}

Let $n \in \N_{>0}$, $b \in \NeciporukSet_{\M}$ and $V \subseteq [n], V \neq
\emptyset$.  According to~(iv), we have $b(m) \leq \alpha \cdot h_\M(m)$
for all $m \in \N, m \geq 4$. Moreover, by
Lemma~\ref{lem:Number_of_subfunctions_upper_bound}, we have $r_V(f_n) \leq
\min\{2^{2^{\card{V}}}, 2^{n - \card{V}}\}$, so since $b$ is non-decreasing,
it follows that
$b(r_V(f_n)) \leq b(\min\{2^{2^{\card{V}}}, 2^{n - \card{V}}\})$.
Now, if $\min\{2^{2^{\card{V}}}, 2^{n - \card{V}}\} \geq x_0$, we get that
\begin{align*}
b(\min\{2^{2^{\card{V}}}, 2^{n - \card{V}}\})
& \leq \alpha \cdot h_\M(\min\{2^{2^{\card{V}}}, 2^{n - \card{V}}\}) &
  \text{by (iv)}\\
& = \alpha \cdot \min\{h_\M(2^{2^{\card{V}}}), h_\M(2^{n - \card{V}})\} &
    \text{by (i)}\\
& = h'_n(\card{V})
\displaypunct{;}
\end{align*}
otherwise (i.e. $\min\{2^{2^{\card{V}}}, 2^{n - \card{V}}\} < x_0$), we get that
\[
b(\min\{2^{2^{\card{V}}}, 2^{n - \card{V}}\}) \leq b(\floor{x_0}) \leq
\alpha \cdot h_\M(\floor{x_0}) = h'_n(\card{V})
\]
since $b$ is non-decreasing and $\floor{x_0} \geq 4$. Hence,
$b(r_V(f_n)) \leq h'_n(\card{V})$ for all $n \in \N_{>0}$,
$b \in \NeciporukSet_{\M}$ and $V \subseteq [n], V \neq \emptyset$.
Therefore, by definition, it follows that for all $n \in \N_{>0}$, we have
\begin{align*}
\NeciporukLB^{\M}_F(n)
& = \max\Bigl\{\sum_{i = 1}^p b(r_{V_i}(f_n)) \mathrel{\Big|}
	       b \in \NeciporukSet_{\M} \text{ and }
	       V_1, \ldots, V_p \text{ partition of } [n]\Bigl\}\\
& \leq \max\Bigl\{\sum_{i = 1}^p h'_n(\card{V_i}) \mathrel{\Big|}
		  V_1, \ldots, V_p \text{ partition of } [n]\Bigr\}\\
& = \max\Bigl\{\sum_{i = 1}^p h'_n(v_i) \mathrel{\Big|}
	       \sum_{i = 1}^p v_i = n \text{ and }
	       \forall i \in [p], v_i>0\Bigr\}
\tag{$\star$} \label{eq:Max_Neciporuk_lower_bound}
\displaypunct{.}
\end{align*}

Let $n \in \N, n \geq \log_2 x_0$ and $v_1, \ldots, v_p\in\N_{>0}$
such that $\sum_{i = 1}^p v_i = n$ that realizes the maximum
\eqref{eq:Max_Neciporuk_lower_bound}. We first show that without loss of
generality we can assume that there exists at most one $j \in [p]$ such that
$\min\{2^{2^{v_j}}, 2^{n - v_j}\} \geq x_0$ and $v_j \leq
 \frac{\log_2 n - 1}{2}$.

If this is not the case then we have $v, v' \in [n]$ such that
$\min\set{2^{2^v}, 2^{n - v}, 2^{2^{v'}}, 2^{n - v'}} \geq x_0$
and $v + v' \leq \log_2 n - 1$. It follows from (iii) and Claim~\ref{claim-h-prime} that
\[
h'_n(v) + h'_n(v') = \alpha \cdot (h_\M(2^{2^v}) + h_\M(2^{2^{v'}})) \leq
\alpha \cdot h_\M(2^{2^{v + v'}}).
\]
But as $v \leq v + v' \leq \log_2 n - 1$, we have
\[
\min\{2^{2^{v + v'}}, 2^{n - (v + v')}\} = 2^{2^{v + v'}} \geq 2^{2^v} =
\min\{2^{2^v}, 2^{n - v}\} \geq x_0
\displaypunct{;}
\]
hence by Claim~\ref{claim-h-prime},
$\alpha \cdot h_\M(2^{2^{v + v'}})= h'_n(v + v')$ and
$h'_n(v) + h'_n(v')\leq h'_n(v + v') $ and the partition that unifies the
corresponding sets would yield a bound at least as big
in~\eqref{eq:Max_Neciporuk_lower_bound}.

If it exists this $j$ is such that
\[
h'_n(v_j) = \alpha \cdot h_\M(2^{2^{v_j}}) \leq
\alpha \cdot h_\M\Bigl(2^{2^{\frac{\log_2 n - 1}{2}}}\Bigr) \leq
\alpha \cdot h_\M(2^n).
\]

Consider now the remaining elements of the partition, i.e. those $i \in [p]
\setminus \{j\}$. If moreover we have $\min\{2^{2^{v_i}}, 2^{n - v_i}\} \geq
x_0$ then by definition of $h'_n$ we have
\begin{equation*}
h'_n(v_i) \leq \alpha \cdot h_\M(2^{n - v_i}) \leq \alpha \cdot h_\M(2^n).
\end{equation*}
As for this case we have  $v_i > \frac{\log_2 n - 1}{2}$ there are at most
$\frac{2 n}{\log_2 n -  1}$ such $i$. Notice that $\frac{n}{\log_2 n -  1} \leq
\frac{3}{2} \cdot \frac{n}{\log_2 n}$ for $n \geq 8$.

If otherwise $\min\{2^{2^{v_i}}, 2^{n - v_i}\} < x_0$, then we have 
$$h'_n(v_i) = \alpha \cdot h_\M(\floor{x_0})$$
and there are at most $n$ such $i$.

Putting all together, we get that
\begin{align*}
\NeciporukLB^{\M}_F(n)
& \leq \max\Bigl\{\sum_{i = 1}^p h'_n(v_i) \mathrel{\Big|}
		  \sum_{i = 1}^p v_i = n \text{ and } \forall i \in [p], v_i >0
	   \Bigr\}\\
& \leq \alpha \cdot h_\M(2^n) +
       \frac{3 n}{\log_2 n} \cdot \alpha \cdot h_\M(2^n) +
       n \cdot \alpha \cdot h_\M(\floor{x_0})\\
& \leq \alpha \cdot h_\M(2^n) +
       \frac{3 n}{\log_2 n } \cdot \alpha \cdot h_\M(2^n) +
       \frac{h_\M(2^n)}{\log_2 n} \cdot n \cdot \alpha \cdot h_\M(\floor{x_0}) & \text{by (ii)}\\
& \leq \alpha \cdot \Bigl(4 + h_\M(\floor{x_0})\Bigr) \cdot
    \frac{n}{\log_2 n} \cdot h_\M(2^n).
\end{align*}
\end{proof}


\section{Upper Bounds for the Computation of \isakl}
\label{sec:upper_bounds}
The \isakl\ functions play a critical role in our approach to studying \Neci's
method.  This section collects size upper bounds for computing \isakl\ on
every model considered in this paper.  These bounds will be required when
limits to the \Neci\ method for these models are investigated.

\begin{theorem}\label{thm:all_uppers}
\label{lem:ISA_LNBP_size_upper_bound}
\label{lem:ISA_BP_size_upper_bound}
\label{lem:ISA_NBP_size_upper_bound}
\label{lem:ISA_LNBF_size_upper_bound}
\label{lem:ISA_BF_size_upper_bound}
Let $\delta \in \N$. For all $k, \ell \in \N_{>0}$,
\begin{align}
    \NBP(\isakl), \PBP(\isakl) &\leq 3 \cdot 2^{k + \frac{\ell}{2}} + 2^\ell\\
\LNBP_\delta(\isakl) &\leq
\begin{cases}
12 \cdot 2^k \max\bigl\{\frac{2^{\ell - \delta}}{\ell - \delta}, \ell\bigr\} +
\frac{2^{2 \ell - \delta}}{\ell - \delta} & \text{if $\ell>\delta$}\\
2^k (3 \ell + 1) + 2 \cdot 2^\ell & \text{if $\ell\leq \delta$}
\end{cases}
\\
\BP(\isakl) &\leq 9 \cdot \frac{2^{k + \ell}}{\ell} + \frac{2^{2
    \ell}}{\ell}\\
\LNBF_\delta(\isakl) &\leq
12 \cdot 2^k \cdot \max\{2^{\ell - \delta}, \ell\} + 3 \cdot 2^\ell
\\
\BF(\isakl) &\leq 7\cdot 2^k\cdot 2^\ell.
\end{align}
\end{theorem}

\begin{proof}
Recall the notation used to refer to the bits of an \isakl\ instance
$\tuple{a}$.
Here we further use $a_1,\ldots,a_k$ for the bits of the primary
pointer and (when relevant) $x_{n+1},\ldots,x_{n+\delta}$ for the
nondeterministic variables.

We begin with simple constructions:

\begin{lemma}\label{lem:constructs}
Let $v_1,\ldots,v_k,y_1,\ldots,y_k,z_1,\ldots, z_{2^k}$ be Boolean
variables, $k\geq 1$. 
\begin{enumerate}
\item A size $2^k-1$ deterministic branching program can ``read''
  $v_1,\ldots,v_k$ and route the $2^k$ possible outcomes to $2^k$
  distinct arcs;
\item A size $3k$ deterministic branching
  program can test whether $v_i=y_i$ holds for every $i\in[k]$;
\item A size $2^{k+1}-2$ deterministic branching program with $2^k$
  distinguished states $s_w$ for $w\in\{0,1\}^k$ can ascertain that
  $(v_1,\ldots,v_k)=w$, i.e., has the property that
  for each $w$, a computation started at
  $s_w$ accepts iff $(v_1,\ldots,v_k)=w$;
\item A size $4k$ formula can test whether $v_i=y_i$ holds for every
  $i\in[k]$; 
\item A formula with leaves $z_1,\ldots,z_{2^k}$ and, for every
  $i\in [k]$, with $2^i$ leaves  $v_i$ or $\neg v_i$
  can compute $z_{\bin_k(v_1,\ldots,v_k) + 1}$.
\end{enumerate}
\end{lemma}

\begin{proof}
For (1), a full binary tree suffices.
For (2), a size-$3$ program can test whether $v_i=y_i$ for a fixed $i$, so a
cascade of $k$ such programs can check equality for every $i$.
For (3), an inverted binary tree first
queries $v_1$ at each of $2^k$ leaves $s_w$, $w\in\{0,1\}^k$; each
answer $a\in\{0,1\}$ branches from $s_w$ to the unique state $s_{w'}$, among
$2^{k-1}$ states at the next level, for which $w=aw'$;
each state at this next level queries $v_2$ and branches to one
of $2^{k-2}$ states at the next level, and so on,
down to level $k$ with two states querying $v_k$, for a
total of $\Sigma_{1\leq i\leq k}2^i$ states; every missing arc in the
above description rejects.

For (4), the formula $\wedge_{1\leq i\leq k}[(v_i\wedge y_i)\vee
(\neg v_i\wedge \neg y_i)]$ expanded into a binary tree has $4k$
leaves.
For (5), we note that
$(\neg v_1\wedge z_1)\vee (v_1\wedge z_2)$ computes
$z_{\bin_1(v_1) + 1}$ and use induction, having computed
$z_{\bin_k(0,v_2,\ldots,v_k) + 1}$ from the leaves
$z_1,\ldots,z_{2^{k-1}}$ and the $2^i$ leaves
  $v_{i+1}$ or $\neg v_{i+1}$ for $i\in[k-1]$, and having computed 
$z_{\bin_k(1,v_2,\ldots,v_k) + 1}$ similarly from the leaves
$z_{2^{k-1}+1},\ldots,z_{2^k}$ and $2^i$ further leaves
  $v_{i+1}$ or $\neg v_{i+1}$ for $i\in[k-1]$.
\end{proof}

\emph{The $\NBP$ case}.
If $\ell=1$ then, by Lemma~\ref{lem:constructs}.1, a
(deterministic) BP of 
size $2^k-1 + 2^k + 2<3\cdot 2^k + 2^\ell$ computes \isakl.
So let $\ell>1$.
For every $w\in \{0,1\}^{\ceiling{\ell/2}}$, $w'\in \{0,1\}^{\floor{\ell/2}}$ and
$p'\in[2^k]$, 
the NBP will have states $s_{(w,w')}$, $(p',s_{w'})$ and $p'$.
Together with further states, the NBP
implements the following:

\begin{itemize}
\item Read bits $a_1,\ldots,a_k,\sec_1,\ldots,\sec_{\ceiling{\ell/2}}$.
\item Guess $w'\in  \{0,1\}^{\floor{\ell/2}}$ and branch
  to $s_{(\sec_1,\ldots,\sec_{\ceiling{\ell/2}},w')}$, forgetting $a_1,\ldots,a_k$.\\
  (For every $w'\in  \{0,1\}^{\floor{\ell/2}}$ and for every $a\in\{0,1\}$, every
  state querying 
  $\sec_{\ceiling{\ell/2}}$, i.e., every bottom node in the binary tree
  formed by the first stage,
  is connected to the state $s_{(\sec_1,\ldots,\sec_{\ceiling{\ell/2}-1},a,w')}$
  with an arc labelled $a$.)
\item If \emph{Data}[$\sec_1,\ldots,\sec_{\ceiling{\ell/2}},w'$] $=1$ then
guess the bits $a'_1,\ldots,a'_k$ of the primary pointer $p\in [2^k]$
and branch to the state $(\bin_k(a'_1,\ldots,a'_k)+1,s_{w'})$.\\ 
(For every $w\in \{0,1\}^{\ceiling{\ell/2}}$ and $w'\in
\{0,1\}^{\floor{\ell/2}}$, the state $s_{w,w'}$ queries \emph{Data}[$w,w'$] and
connects via an arc labelled $1$ to every state $(p',s_{w'})$,
$p'\in[2^k]$.)
\item Ascertain that $w'$ was guessed correctly.\\
(For each $p'\in [2^k]$ separately, apply Lemma~\ref{lem:constructs}.3 to the
distinguished states $(p',s_{w'})$, $w'\in\{0,1\}^{\floor{\ell/2}}$, to
ascertain that a computation from $(p',s_{w'})$ reaches the state $p'$
iff $(\sec[p']_{\ceiling{\ell/2}+1},\ldots,\sec[p']_\ell)=w'$.)
\item Ascertain that $p$ was guessed correctly.\\
(Apply Lemma~\ref{lem:constructs}.3, to the $2^k$ distinguished states
$p'$, to ascertain that $(a_1,\ldots,a_k)=(a_1',\ldots,a_k')$.)

\end{itemize}
The first stage uses $2^{k+\ceiling{\ell/2}}-1$ states by
  Lemma~\ref{lem:constructs}.1.
The second stage needs the $2^\ell$ states $s_{w,w'}$.
The fourth stage uses $2^k$ times $2^{\floor{\ell/2}+1}-2$ states by
Lemma~\ref{lem:constructs}.3 (and also includes the $2^{k+\floor{\ell/2}}$ states
$(p',s_{w'})$).
The last stage uses 
$2^{k+1}-2$ states by 
Lemma~\ref{lem:constructs}.3 for a total 
$<2^k(2^{\ceiling{\ell/2}} + 2\cdot 2^{\floor{\ell/2}}-2)+2^{k+1} + 2^\ell$, which equals
$2^k(3\cdot 2^{\ell/2}) + 2^\ell$ when $\ell$ is
even and $2^k(\frac{4}{\sqrt{2}}\cdot 2^{\ell/2}) + 2^\ell<2^k(3\cdot
2^{\ell/2}) + 2^\ell$ when $\ell$ is odd.

\emph{The $\PBP$ case}.  It is easy to check that the above NBP has a 
unique accepting path for any input for which $\isakl$ is 1 and hence as 
a $\oplus$BP it also computes $\isakl$.

\emph{The $\LNBP_\delta$ case}.
If $\ell \leq \delta$ then
the secondary \isakl\ pointer is no wider than $\delta$, i.e.,
contains no more than $\delta$ bits.
So a $\delta$-LNBP can ``store'' the secondary pointer within its
first $\ell$ nondeterministic variables $x_{n+1},\ldots,x_{n+\ell}$  and
solve \isakl\ as follows:
\begin{itemize}
\item Read the primary pointer
\item Check that $(\sec_1,\ldots,\sec_\ell)=(x_{n+1},\ldots,x_{n+\ell})$
\item Forget everything
\item Read $x_{n+1},\ldots,x_{n+\ell}$
\item Check that data[$x_{n+1},\ldots,x_{n+\ell}$]=1.
\end{itemize}

The first and second steps use $2^k-1$ and $2^k 3 \ell$ states
respectively,  appealing to 
Lemma~\ref{lem:constructs}.1 and Lemma~\ref{lem:constructs}.2.
Note that across the second step, neither the secondary pointer nor
$x_{n+1}, \ldots, \allowbreak x_{n+\ell}$ are remembered.
The third step merges every arc that survived
the second step and thus requires no state.
The fourth and fifth steps require $2^\ell-1$ and $2^\ell$ states, for
a total $<2^k + 2^k 3 \ell + 2 \cdot 2^\ell$.

Now suppose that $\ell > \delta$, i.e.,
the secondary pointer is strictly wider than $\delta$.
Let $m \in [\ell - \delta
- 1]$, to be set optimally later.
A $\delta$-LNBP can implement the following strategy, where
grey-shaded regions in the diagrams indicate the portion of the
\isakl\ variables that are remembered, at exponential cost in
numbers of states, at any given time.

\begin{enumerate}
\item Read the primary pointer:

\centerline{\includegraphics[height=0.12\textheight]{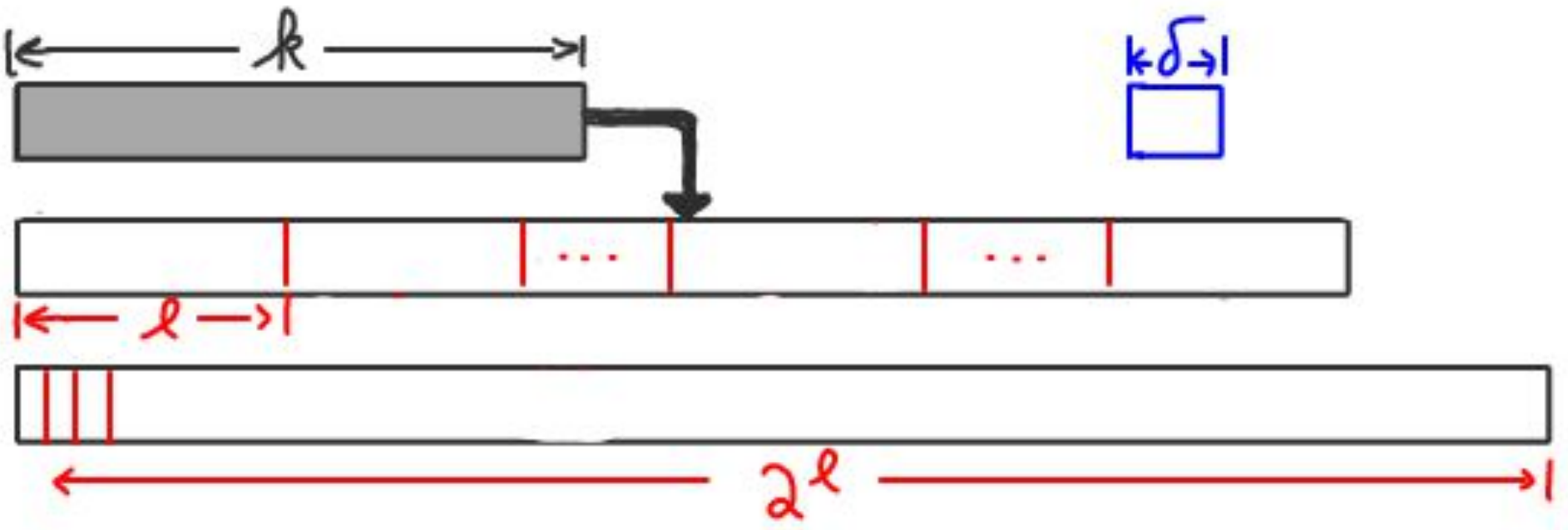}}
Uses $2^k-1$ states as per Lemma~\ref{lem:constructs}.1.

\item Check $\delta$ contiguous secondary pointer bits for equality with
  $x_{n+1},\ldots,x_{n+\delta}$:

\centerline{\includegraphics[height=0.12\textheight]{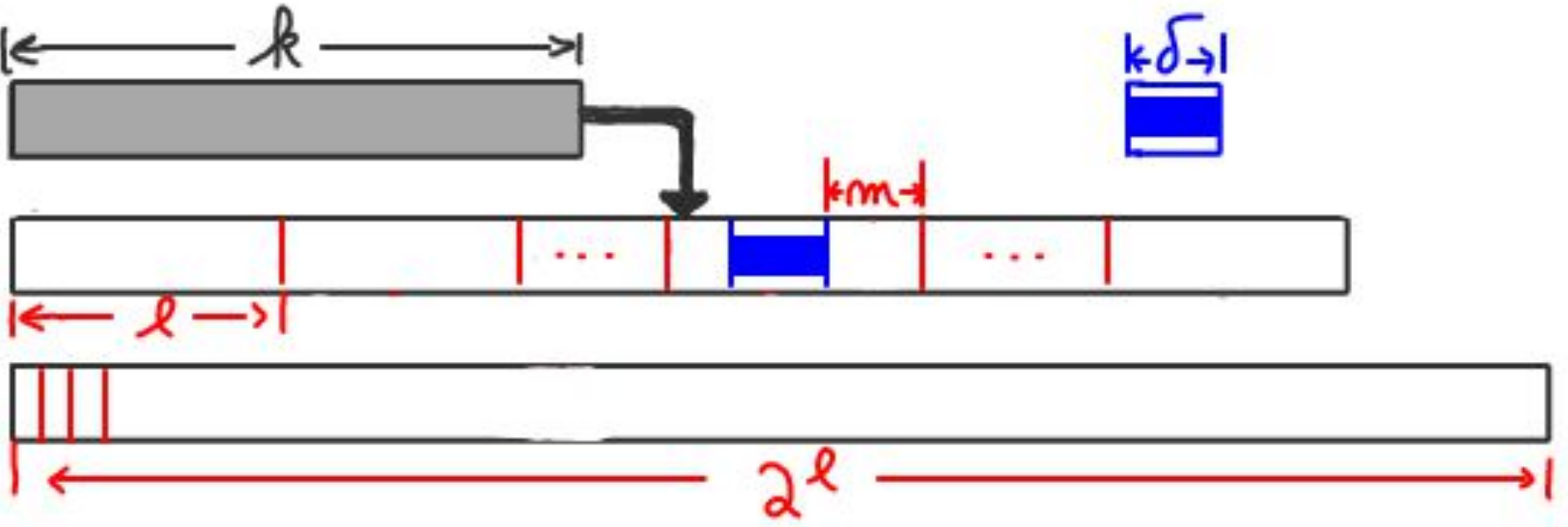}}
Uses $2^k$ times $3\delta$ states, again by
Lemma~\ref{lem:constructs}.2. None of the checked bits are
remembered.

\item Read $\ell-m-\delta$ other contiguous bits from the secondary pointer:

\centerline{\includegraphics[height=0.12\textheight]{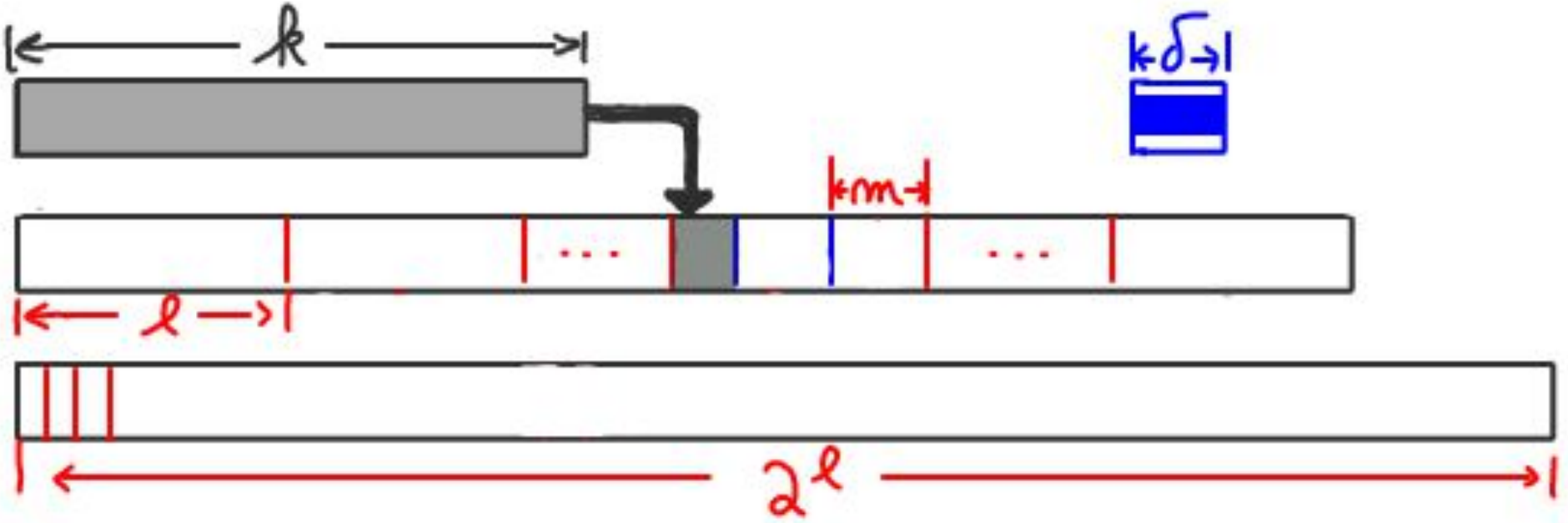}}
Uses $2^k$ times $(2^{\ell-m-\delta}-1)$ states.

\item Forget the primary pointer:

\centerline{\includegraphics[height=0.12\textheight]{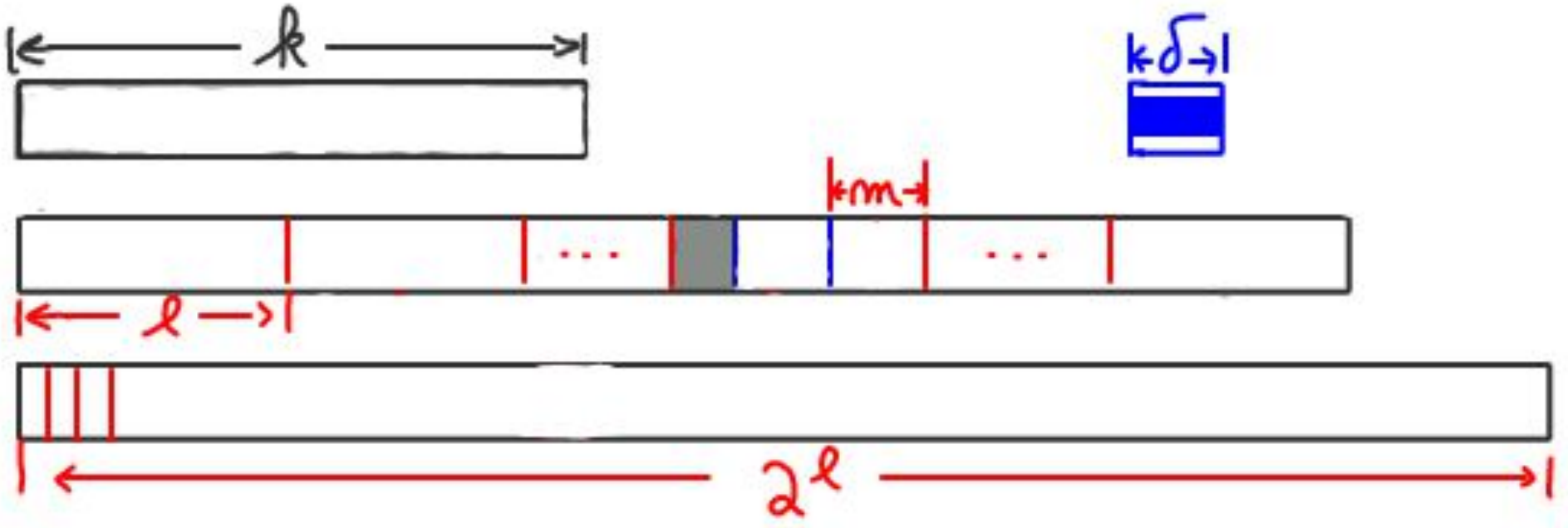}}
No state required.

\item Read and remember the nondeterministic bits:

\centerline{\includegraphics[height=0.12\textheight]{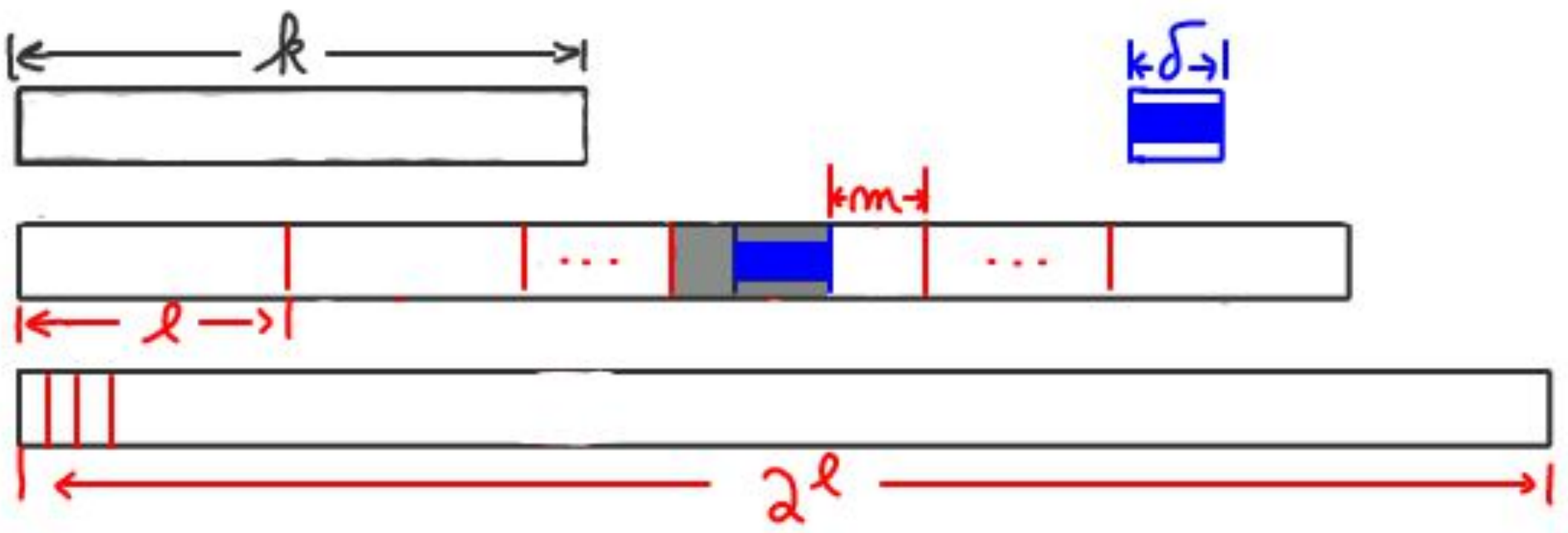}}
Uses $2^{\ell-m-\delta}(2^{\delta}-1)<2^{\ell-m}$ states.

\item Read the data bits that remain candidates:

\centerline{\includegraphics[height=0.12\textheight]{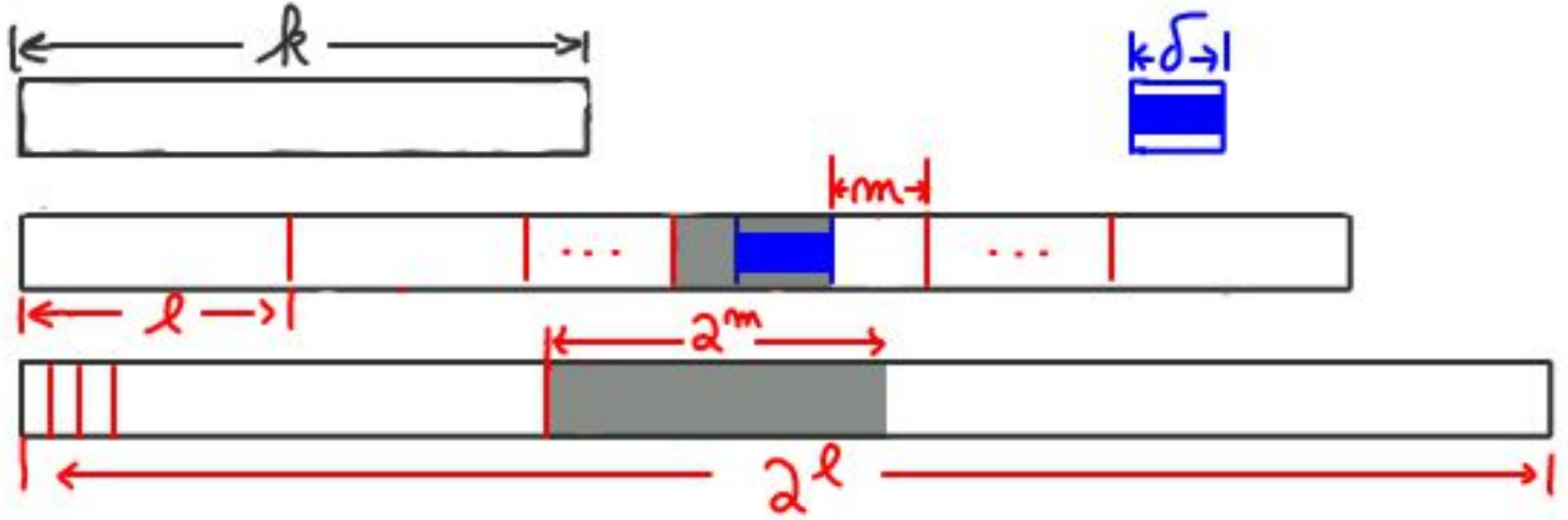}}
Uses $2^{\ell-m}(2^{2^m}-1)< 2^{\ell-m}2^{2^m}$ states.

\item Forget the part of the secondary pointer that was read:

\centerline{\includegraphics[height=0.12\textheight]{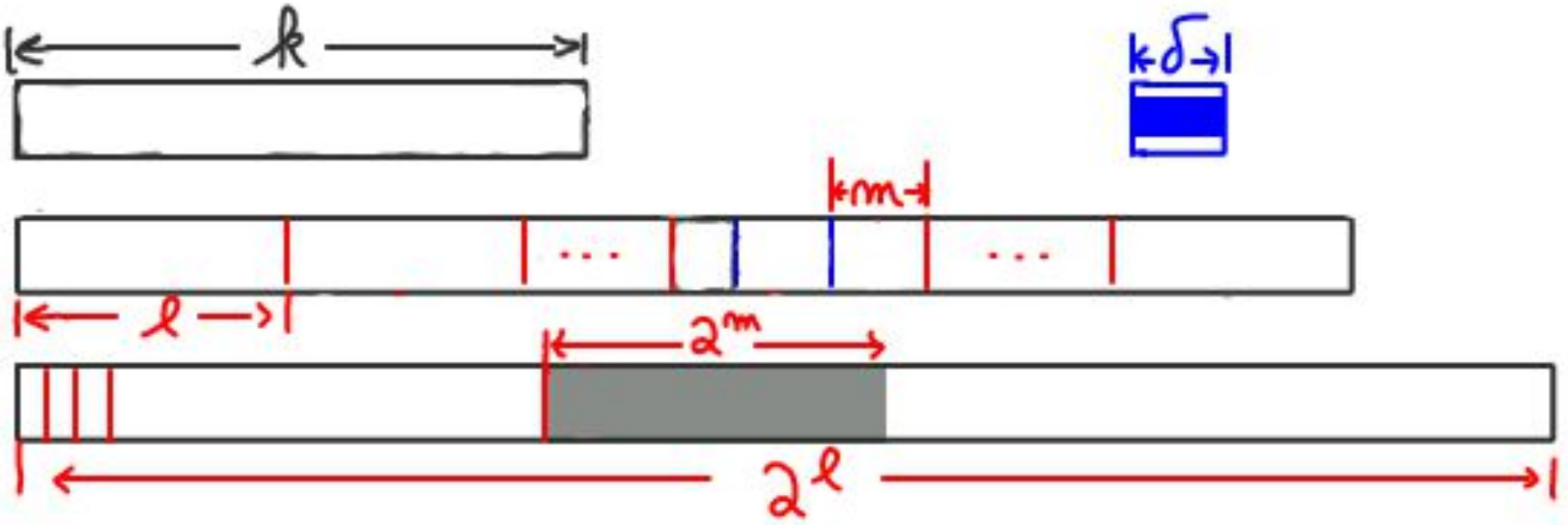}}
No state required.

\item Read the primary pointer:

\centerline{\includegraphics[height=0.12\textheight]{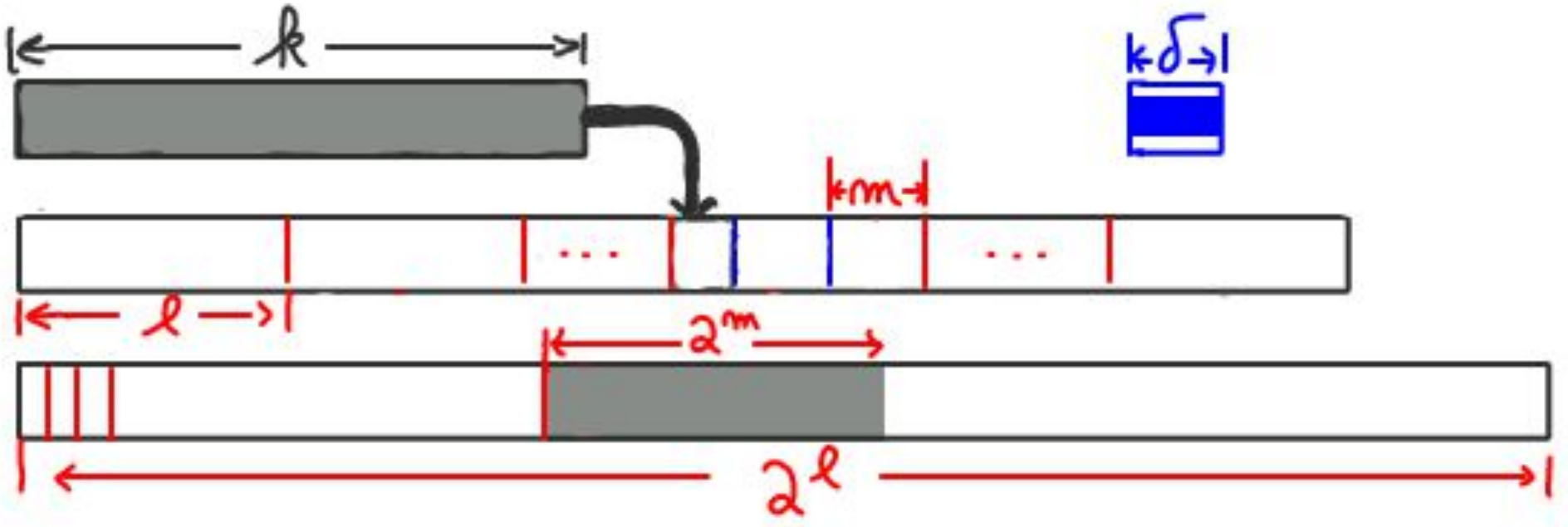}}
Uses $2^{2^m}(2^k-1)<2^k2^{2^m}$ states.

\item Read the secondary pointer bits that were never yet accessed:

\centerline{\includegraphics[height=0.12\textheight]{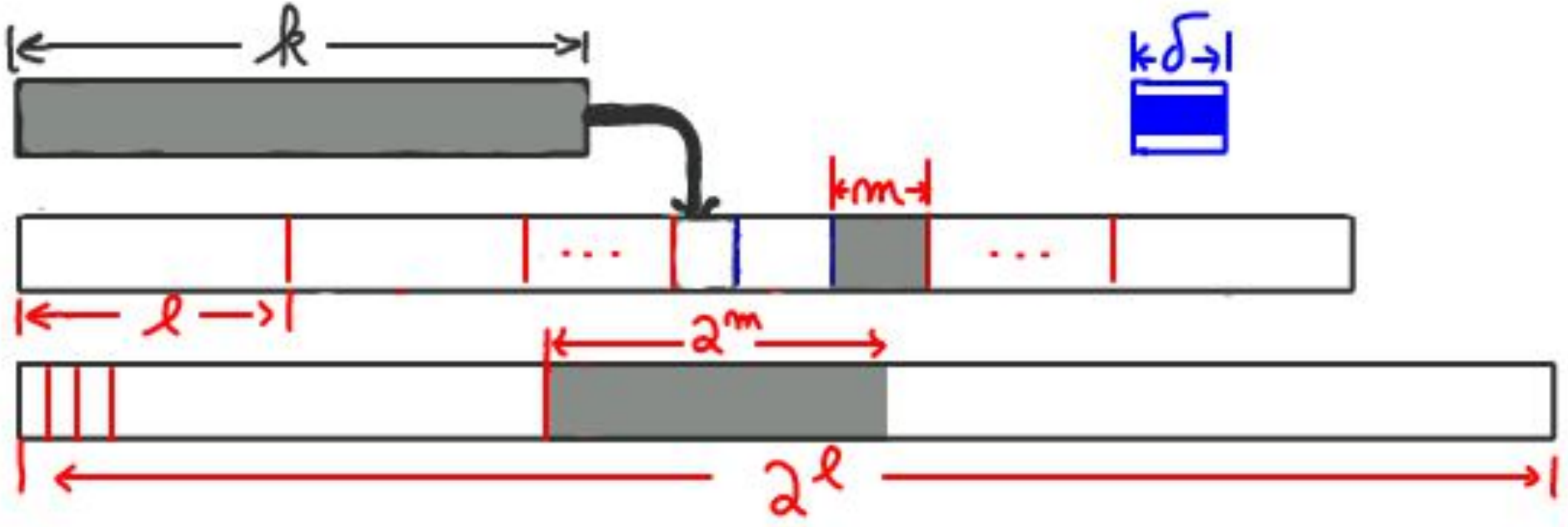}}
Uses $2^{2^m}2^k(2^m-1)<2^k2^m2^{2^m}$ states.

\item Output the appropriate data bit from memory: no state required.
\end{enumerate}
The resulting $\delta$-LNBP has fewer than
$$2^k + 2^k 3 \delta + 2^k 2^{\ell - \delta - m} + 2^{\ell - m} +
  2^{\ell - m} 2^{2^m} + 2^k 2^{2^m} + 2^k2^m 2^{2^m}$$
states, which is less than
$$12 \cdot 2^k \max\Bigl\{\frac{2^{\ell - \delta}}{\ell - \delta}, \ell\Bigr\} +
\frac{2^{2 \ell - \delta}}{\ell - \delta}$$
when $m$ is set to 
$\floor{\log_2(\ell - \delta - \log_2((\ell - \delta)^2))}$ and
$\ell - \delta \geq 8$ (and the degenerate case in which
$1 \leq \ell - \delta < 8$ is treated separately by using a simpler method to
compute $\isakl$).

\emph{The $\BP$ case}.
Follows from the $\LNBP_\delta$ case by setting $\delta=0$. More
specifically, stages 2 and 5 in the construction of the $\delta$-NLBP
are skipped.

\emph{The $\delta$}-\emph{$\LNBF$ case}.
We will not exploit more than $\ell$ nondeterministic variables amongst
$x_{n+1}, \ldots \allowbreak, x_{n+\delta}$ so we suppose that
$\delta\leq \ell$.
Let $m=\ell-\delta$.
The nondeterministic formula $V\wedge D$ solves \isakl\ provided that
$V$ and $D$ fulfil
\begin{align*}
V=1\ &\text{iff } (\sec_{m+1},\ldots,\sec_\ell) =
(x_{n+1},\ldots,x_{n+\delta}),\\
D=1\ &\text{iff }
\textit{Data}[F_1,\ldots,F_m,x_{n+1},\ldots,x_{n+\delta}]=1,
\end{align*}
and for $1\leq j\leq m$, $F_j$ evaluates to $\sec_j$.
By Lemma~\ref{lem:constructs}.5, $D$ exists such that
\begin{align}
|D| &= 2^\ell + \Sigma_{j=1}^m 2^j \cdot |F_j|
+ \Sigma_{j=m+1}^\ell 2^j
 <
3\cdot 2^\ell + \Sigma_{j=1}^m 2^j \cdot |F_j|.
 \label{eqn:VD}
\end{align}
By Lemma~\ref{lem:constructs}.5, each
formula $F_j$, $1\leq j\leq m$, can be constructed of size
\begin{align}
|F_j| &= 2^k + \Sigma_{j=1}^k 2^j < 3\cdot 2^k. \label{eqn:F}
\end{align}
By Lemma~\ref{lem:constructs}.4, for every $p\in[2^k]$, a formula
$V_p$ of size $4\delta$ can be constructed that evaluates to $1$ iff
$(\sec[p]_{m+1},\ldots,\sec[p]_\ell) = 
(x_{n+1},\ldots,x_{n+\delta})$.
The formula $V$ can then be constructed using
Lemma~\ref{lem:constructs}.5, taking $z_1,\ldots,z_{2^k}$ as
$V_1,\ldots,V_{2^k}$. The size of $V$ is then
\begin{align}
|V| &=\Sigma_{j=1}^{2^k} |V_j|\ +\ \Sigma_{j=1}^k 2^j <
2^k\cdot 4\delta +  2^{k+1} \leq 2^k \cdot 6\ell. \label{eqn:V} 
\end{align}
Substituting (\ref{eqn:F}) into (\ref{eqn:VD}) and using
(\ref{eqn:V}), the size of $V\wedge D$ is at most
\begin{align*}
2^k \cdot 6\ell + 3\cdot 2^\ell + 2^{m+1} 3\cdot 2^k \leq
6\cdot 2^k (\ell + 2^m) + 3\cdot 2^\ell \leq 
12\cdot 2^k \cdot \max\{2^m,\ell\} + 3\cdot 2^\ell.
\end{align*}

\emph{The $\BF$ case}.
Follows from the $\delta$-$\LNBF$ case by setting $\delta = 0$. More
sharply, $V$ from that construction is not needed, and $|D|=2^\ell +
\Sigma_{j=1}^\ell 2^j|F_j| < 2^\ell + 3\cdot 2^k \cdot 2^{\ell+1} < 7\cdot
2^k\cdot 2^\ell$.
\end{proof}

\section{Nondeterministic and Parity Branching Programs revisited}
We note in this section that, in the case of $\NBP$ and $\PBP$, the flexibility
added by Definition~\ref{defi-neci} over Definition~\ref{def:Weak_Neciporuk_NBP}
yields no better lower bounds.

We first define the function $b_{\NBP,\PBP}\colon \N_{>0} \rightarrow \N$ by 
\[
b_{\NBP,\PBP}(m)=
\begin{cases}
\ceiling{\sqrt{\frac{1}{2} \log_2 m} - 1} & \text{if $m \geq 4$}\\
0 & \text{otherwise}
\end{cases}
\]
for all $m \in \N_{>0}$.
Using the same strategy as in the proofs of
Lemma~\ref{lem:Weak_Neciporuk_NBP_principle} and
Proposition~\ref{ptn:Weak_Neciporuk_NBP_size_lower_bound}, we can prove the
following.

\begin{proposition}
\label{ptn:Neciporuk_NBP_size_lower_bound}
$b_{\NBP,\PBP}$ is a \Neci bounding function for  the $\NBP$ (respectively, $\PBP$) size complexity measure;
i.e., $b_{\NBP,\PBP} \in \NeciporukSet_{\NBP},\NeciporukSet_{\PBP}$.
\end{proposition}

Combining this with Lemmas~\ref{ED-subfunctions-lemma}
and~\ref{lem:ISA_partitioning}, we can
immediately derive asymptotic lower bounds on
$\NBP(\ED_n)$ and $\NBP(\ISA_n)$ using \Neci's method and
hence on
$\NeciporukLB^{\NBP}_{\ED}$,
$\NeciporukLB^{\PBP}_{\ED}$,
$\NeciporukLB^{\NBP}_{\ISA}$,
and
$\NeciporukLB^{\PBP}_{\ISA}$.

\begin{proposition}
\label{ptn:ISA_NBP_size_lower_bound}
$\NeciporukLB^{\NBP}_{\ED}(n), \NeciporukLB^{\NBP}_{\ISA}(n),
\NeciporukLB^{\PBP}_{\ED}(n), \NeciporukLB^{\PBP}_{\ISA}(n) \in
 \Omega\Bigl(\frac{n^{3/2}}{\log_2 n}\Bigr)$
 and hence
$\NBP(\ED_n)$, $\NBP(\ISA_n)$, $\PBP(\ED_n)$, $\PBP(ISA_n)$ are all 
$\Omega\Bigl(\frac{n^{3/2}}{\log_2 n}\Bigr)$.
\end{proposition}

\begin{sloppypar}
Then, we can show that $b_{\NBP,\PBP}$ is in fact the asymptotically largest
function
in $\NeciporukSet_{\NBP}\cup \NeciporukSet_{\PBP}$ and that the previous lower
bound is in fact also the
asymptotically largest we may obtain.
To do this, we appeal to our upper bound from Theorem~\ref{thm:all_uppers} on
the size of NBPs and $\oplus$BPs computing $\isakl$ and apply
Lemma~\ref{lem:Neciporuk's_method_function_limitation_meta-result}.
\end{sloppypar}

\begin{proposition}
\label{ptn:Neciporuk's_method_function_limitation_NBP_size}
There exists a constant $c \in \R_{>0}$ verifying that any $b \in
\NeciporukSet_{\NBP}\cup \NeciporukSet_{\PBP}$ is such that
$b(m) \leq c \cdot b_{\NBP,\PBP}(m)$, for $m\geq 4$.
\end{proposition}

\begin{proof}
Let $g\colon \cointerval{1}{+\infty} \to \R_{\geq 0}$ be the non-decreasing
function defined by $g(x) = 4 \cdot 2^{\frac{3}{2} x}$ for all
$x \in \cointerval{1}{+\infty}$.
Theorem~\ref{lem:ISA_NBP_size_upper_bound} tells us that for all $k \in \N_{>0}$,
we have
\[
\NBP(\isakk), \PBP(\isakk) \leq 3 \cdot 2^{\frac{3}{2} k} + 2^k \leq
4 \cdot 2^{\frac{3}{2} k} = g(k)
\]
and moreover, $\frac{g(k + 1)}{g(k)} = \frac{4 \cdot 2^{\frac{3}{2} (k + 1)}}{4
  \cdot 2^{\frac{3}{2} k}} = 2 \sqrt{2}$ for all $k \in \N_{>0}$. Therefore, by
Lemma~\ref{lem:Neciporuk's_method_function_limitation_meta-result}, any $b \in
\NeciporukSet_{\NBP}\cup \NeciporukSet_{\PBP}$ verifies
\[
b(m) \leq 2 \sqrt{2} \cdot \frac{g(\log_2\log_2 m)}{\log_2 m} =
2 \sqrt{2} \cdot \frac{4 \cdot 2^{\frac{3}{2} \log_2\log_2 m}}{\log_2 m} =
8 \sqrt{2} \cdot \sqrt{\log_2 m}
\]
for all $m \in \N, m \geq 4$.
\end{proof}

Finally, using this and
Lemma~\ref{lem:Neciporuk's_method_limitation_complexity_lower_bound_meta-result},
we get the following result, showing that the asymptotically greatest lower
bound we may expect using \Neci's method for $\NBP$ is (asymptotically)
equivalent to the lower bound for $\ISA$ given in
Proposition~\ref{ptn:ISA_NBP_size_lower_bound}.

\begin{theorem}
\label{thm:Neciporuk's_method_limitation_NBP_size_lower_bound}
For any family of Boolean functions $F = \{f_n\}_{n \in \N}$,
$\NeciporukLB^{\NBP}_F(n),
\NeciporukLB^{\PBP}_F(n) \in
 \Omicron\Bigl(\frac{n^{3/2}}{\log_2 n}\Bigr)$.
\end{theorem}

\begin{proof}
  We aim at applying
  Lemma~\ref{lem:Neciporuk's_method_limitation_complexity_lower_bound_meta-result}
  which requires four hypotheses, (i) to~(iv).

For (i), let $h\colon \cointerval{4}{+\infty} \to \R_{\geq 0}$ be the function
defined by $h(x) = \sqrt{\log_2 x}$ for all $x \in \cointerval{4}{+\infty}$ and
$x_0 = 2^8$; as required, $h$ is non-decreasing on $\cointerval{2^8}{+\infty}$.

For (ii), notice that $h(2^x)=\sqrt{x}\geq \log_2 x$ for all
$x \in \cointerval{2^8}{+\infty}$.

For (iii), for all $v, v' \in \N$ verifying  $2^{2^v} \geq 2^8$ and
$2^{2^{v'}} \geq 2^8$, we have
$h(2^{2^v}) + h(2^{2^{v'}}) = \sqrt{2^v} + \sqrt{2^{v'}} \leq
 \sqrt{2^{v + v'}} = h(2^{2^{v + v'}})$ because 
$x + y \leq x y$ when $x,y \geq 2$.

For (iv), by
Proposition~\ref{ptn:Neciporuk's_method_function_limitation_NBP_size}, we know
that any $b \in \NeciporukSet_{\NBP}\cup \NeciporukSet_{\PBP}$
is such that $b(m) \leq \alpha \cdot
\sqrt{\log_2 m} = \alpha \cdot h(m)$ for all $m \in \N, m \geq 4$.

We can therefore apply
Lemma~\ref{lem:Neciporuk's_method_limitation_complexity_lower_bound_meta-result}
with $x_0 = 2^8$ and get that for any family of
Boolean functions $F = \{f_n\}_{n \in \N}$ and all $n \in \N, n
\geq 8$,
\[
\NeciporukLB^{\NBP}_F(n),
\NeciporukLB^{\PBP}_F(n) \leq
\alpha \cdot \bigl(4 + h(\floor{2^8})\bigr) \cdot
\frac{n}{\log_2 n} \cdot h(2^n) =
c \cdot \frac{n}{\log_2 n} \cdot \sqrt{n} =
c \cdot \frac{n^{3/2}}{\log_2 n}
\displaypunct{,}
\]
which implies that
$\NeciporukLB^{\NBP}_F(n),
\NeciporukLB^{\PBP}_F(n) \in
 \Omicron\Bigl(\frac{n^{3/2}}{\log_2 n}\Bigr)$.
\end{proof}

\section{Deterministic and Limited Nondeterministic Branching Programs}
\label{sec:Limited_nondeterministic_branching_programs}
In this section, we focus on the model of Boolean deterministic branching
programs, as well as its limited nondeterministic counterpart. In the case of
$\BP$, results related to the \Neci method have been well-known  for a long time
(see for instance \cite[Chapter 14, Section 3]{we87} or \cite{alzw89}).
Reproving these results using what we presented in Section~\ref{sec:neciporuk}
is an opportunity to confirm the usability and validity of our approach.

Concerning limited nondeterministic branching programs, the definition of the
model itself, as well as the results presented in this section concerning
\Neci's method for the associated measure seem to be novel.

For all $\delta \in \N$, let us define the functions
$b_{\LNBP_\delta}\colon \N_{>0} \to \N$ and $b_{\BP}\colon \N_{>0} \to \N$ given
by 
\[
b_{\LNBP_\delta}(m)  =
\begin{cases}
\ceiling{\frac{1}{6} h_{\LNBP_\delta}(m)} & \text{if $m \geq 4$}\\
0 & \text{otherwise}
\end{cases}
\]
and
\[
b_{\BP}(m) =
\begin{cases}
\ceiling{\frac{1}{6}\frac{\log_2 m }{\log_2 \log_2 m}} & \text{if $m \geq 4$}\\
0 & \text{otherwise}
\end{cases}
\]
for all $m \in \N_{>0}$, where
$h_{\LNBP_\delta}\colon \cointerval{4}{+\infty} \to \R$ is defined as
\[
h_{\LNBP_\delta}(x) =
\begin{array}[t]{@{}r@{\,\,}l@{\,\,}l@{}}
\begin{cases}
\max\bigl\{\frac{\log_2 x}{2^\delta (\log_2 \log_2 x - \delta)},
	   \log_2 \log_2 x\bigr\} & \text{if $2^{2^{\delta + 1}} \leq x$}\\
\log_2 \log_2 x & \text{otherwise}
\displaypunct{.}
\end{cases}
\end{array}
\]
It is straightforward to see that $b_{\BP}(m) \leq b_{\LNBP_0}(m)$ for all
$m \in \N_{>0}$ and that equality holds as soon as
$b_{\BP}(m)\geq \log_2\log_2 m$.

To prove that $b_{\BP} \in \NeciporukSet_{\BP}$, we use the well known idea
that is usually used (see for instance \cite{we87}, \cite{alzw89} or
\cite{ju12}) to derive a specific function $b \in \NeciporukSet_{\BP}$, which
is the fact that, given a Boolean function $f$ and a Boolean BP $P$ that
computes it, we can compute any subfunction $f|_\rho$ of $f$ with a Boolean BP
obtained from $P$ by ``fixing'' the values of the variables to which a value is
affected by $\rho$ (removing the associated vertices and directly linking their
predecessors to their successors through the arcs labelled accordingly).
Therefore, if we denote by $s$ the number of vertices in $P$ labelled by
elements from $V$, we get that an upper bound on the maximum number of
subfunctions computed by BPs with $s$ vertices obtained by ``fixing'' the
values of a given set of variables in a given BP implies a lower bound on $s$
depending on $r_V(f)$, as this number must be at least as big as $r_V(f)$.  For
the case of limited nondeterminism, it suffices to observe that a Boolean
$\delta$-LNBP (for $\delta \in \N$) computing some Boolean function $f$, does
in fact deterministically compute a proof-checker function $g$ for $f$.  We can
then combine the aforementioned technique with
Lemma~\ref{lem:Number_of_subfunctions_proof_checker_upper_bound} binding the
number of subfunctions of $f$ on $V$ and the number of subfunctions of $g$ on
$V$.

\begin{proposition}
\label{ptn:Neciporuk_LNBP_size_lower_bound}
\label{ptn:Neciporuk_BP_size_lower_bound}
$b_{\LNBP_\delta} \in \NeciporukSet_{\LNBP_\delta}$ for all $\delta \in \N$.
In particular, $b_{\BP} \in \NeciporukSet_{\BP}$.
\end{proposition}

\begin{proof}
Let $\delta \in \N$. It is not
too difficult to show that $b_{\LNBP_\delta}$ and $b_{\BP}$ are non-decreasing,
we leave this to the reader.

Let $f$ be a $n$-ary Boolean function on $V$ and $V_1, \ldots, V_p$ a partition of
$V$.  Let $P$ be a Boolean $\delta$-LNBP computing $f$ and let $g$ be the $(n
+ \delta)$-ary Boolean function computed by $P$ when considering the $\delta$
nondeterministic bits as regular input variables (that is, $g$ is such that,
for all $\tuple{a} \in \{0, 1\}^{V \cup [\delta]}$, $g(\tuple{a}) = 1$ if, and only
if, $P[\tuple{a}]$ contains a path from $s$ to $t_1$). $g$ is a proof-checker
function for $f$.

For all $i \in [p]$ we will denote by $s_i \in \N$ the number of vertices in $P$
labelled by elements in $V_i$, as well as $q \in \N$ the number of vertices
labelled by elements in $U = [\delta]$. It is clear that $P$
is of size $\sum_{i = 1}^p s_i + q \geq \sum_{i = 1}^p s_i$.

We now claim that $s_i \geq b_{\LNBP_\delta}(r_{V_i}(f))$ for all $i \in
[p]$.

Let $i \in [p]$. Let $V_i'$ be the subset of $V_i$ containing all indices of
variables on which $f$ depends. Then, by
Lemma~\ref{lem:Number_of_subfunctions_with_variable_independence}, $r_{V_i}(f)
= r_{V_i'}(f)$. Moreover for each element $l \in V_i'$, $P$ contains at least
one vertex labelled by $l$. By
Lemma~\ref{lem:Number_of_subfunctions_upper_bound}, it follows that $r_{V_i}(f)
= r_{V_i'}(f) \leq 2^{2^{\card{V_i'}}} \leq 2^{2^{s_i}}$ and $s_i \geq
\log_2 \log_2(r_{V_i}(f))$.

If $r_{V_i}(f) \leq 3$ the claim is obvious from the
definition of $b_{\LNBP_\delta}$.

In the case where $4 \leq r_{V_i}(f) <
4^{2^\delta}$ we have $\ceiling{\frac{1}{6} \log_2\log_2(r_{V_i}(f))} =
b_{\LNBP_\delta}(r_{V_i}(f))$ and we are also done as $s_i$ is an integer and
$s_i \geq \log_2\log_2(r_{V_i}(f))$.

We now assume $r_{V_i}(f) \geq 4^{2^\delta}$. In particular this implies that $s_i \geq 1$.

Observe that for all $h\colon \{0, 1\}^{V_i} \to \{0, 1\}$ a subfunction of $g$
on $V_i$, by definition, there exists a partial assignment $\rho \in \{0,
1\}^{V\setminus V_i \cup [\delta]}$ such that $g|_\rho = h$, so it is not too
difficult to see that $h$ is computed by the Boolean BP of size $s_i$ obtained
from $P$ by:
\begin{enumerate}
    \item
	removing all non sink vertices labelled by variables not in $V_i$;
      \item defining the new start vertex as the only vertex whose label is in
        $V_i$ and connected to the start vertex of $P$ by a path of nodes
        labelled by a variable outside of $V_i$ and arcs labelled
        consistently with $\rho$;
      \item connecting a vertex $u$ to a vertex $v$ by an arc labelled by $a \in
        \set{0, 1}$ if, and only if, there exists a path from $u$ to $v$ in $P$
        verifying that any intermediate vertex of the path is labelled by a
        variable outside of $V_i$, the first arc is labelled by $a$ and each arc
        (but the first one) is labelled consistently with $\rho$.
\end{enumerate}
Thus, $r_{V_i}(g)$ is necessarily upper-bounded by the number of syntactically
distinct such BPs we can build from $P$ that way. Since, for such a BP, there
are at most $s_i + 2$ possible choices for the start vertex and by functionality
of the set of arcs labelled $0$ and the set of arcs labelled $1$ seen as
successor relations, there are at most $(s_i + 1)^{s_i}$ possible choices for
the set of arcs labelled $0$, as well as at most $(s_i + 1)^{s_i}$ possible
choices for the set of arcs labelled $1$, $r_{V_i}(g)$ is at most
$(s_i + 2) (s_i + 1)^{2 s_i}$.
Assuming $2 \leq s_i$ we get:
\begin{align*}
r_{V_i}(g) \leq & (s_i + 2) (s_i + 1)^{2 s_i} && = 2^{\log_2(s_i + 2) + 2 s_i \log_2(s_i + 1)}\\
&&& \leq 2^{3 s_i \log_2(s_i+2)} \\
&&& \leq 2^{6 s_i \log_2(s_i)} & \text{as $2 \leq s_i$}
\displaypunct{.}
\end{align*}
It follows that $s_i \geq \frac{1}{6} \cdot
	  \frac{\log_2(r_{V_i}(g))}{\log_2\log_2(r_{V_i}(g))}$. If $s_i=1$ it
          is clear as $r_{V_i}(g)$ is then at most $4$, and if $s_i\geq 2$ we
          would otherwise  have
\begin{align*}
s_i \log_2(s_i)
& < \frac{1}{6} \cdot \frac{\log_2(r_{V_i}(g))}{\log_2\log_2(r_{V_i}(g))}
    \log_2\Bigl(\frac{1}{6} \cdot \frac{\log_2(r_{V_i}(g))}
				       {\log_2\log_2(r_{V_i}(g))}\Bigr)\\
& = \frac{\log_2(r_{V_i}(g))}{6} -
    \frac{\log_2(r_{V_i}(g)) \log_2\bigl(6 \log_2\log_2(r_{V_i}(g))\bigr)}
	 {6 \log_2\log_2(r_{V_i}(g))}\\
& < \frac{\log_2(r_{V_i}(g))}{6}
\end{align*}
(observe that the last inequality follows from the fact that the subtracted
member must necessarily be positive since $r_{V_i}(g) \geq 4$).
From Lemma~\ref{lem:Number_of_subfunctions_proof_checker_upper_bound} we
have $r_{V_i}(g) \geq r_{V_i}(f)^{\frac{1}{2^\delta}}$. The function
$\frac{\log_2(x)}{\log_2\log_2(x)}$ being non-decreasing on
$\cointerval{e^{e \ln(2)}}{+\infty}$ and as
$\frac{\log_2(x)}{\log_2\log_2(x)} \leq 2$ for
$x \in \ccinterval{4}{e^{e \ln(2)}}$, we get for
$r_{V_i}(f)^{\frac{1}{2^\delta}} \geq 4$:

\[
s_i \geq
\frac{1}{6} \cdot \frac{\log_2(r_{V_i}(f))}
		       {2^\delta \bigr(\log_2\log_2(r_{V_i}(f)) - \delta\bigl)}
\displaypunct{.}
\]

In conclusion, for all $\delta,n \in \N$, and any $n$-ary Boolean function $f$
on $V$ and any partition $V_1, \ldots, V_p$ of $V$, it
holds that $\LNBP_\delta(f) \geq \sum_{i = 1}^p b_{\LNBP_\delta}(r_{V_i}(f))$,
hence $b_{\LNBP_\delta} \in \NeciporukSet_{\LNBP_\delta}$.  It also directly
follows that $b_{\BP} \in \NeciporukSet_{\BP}$ because $b_{\BP}$ is
non-decreasing and $b_{\BP}(m) \leq b_{\LNBP_0}(m)$ for all $m \in \N_{>0}$.
\end{proof}

Let us define
$\Xix_{\LNBP}\colon [2,+\infty) \times \N_{>0}\rightarrow\R$ by
\[
	 \Gamma(x, \delta)=
	 \begin{cases}
	  \max\bigl\{\frac{x^2}{2^\delta (\log_2 x - \delta) \log_2 x},
		     x\bigr\} & \text{if $2^{\delta + 1} \leq x$}\\
	  x & \text{otherwise}
	  \displaypunct{.}
	  \end{cases}
\]
Using the previous proposition and Lemma~\ref{lem:ISA_partitioning}, we can
immediately derive the following asymptotic lower bound on
$\NeciporukLB^{\LNBP_{\Delta(n)}}_{\ISA}$ for any $\Delta\colon \N \to \N$.

\begin{proposition}
\label{ptn:ISA_LNBP_size_lower_bound}
\label{ptn:ISA_BP_size_lower_bound}
$\NeciporukLB^{\LNBP_{\Delta(n)}}_{\ISA}(n) \in
 \Omega\bigl(\Xix_{\LNBP}(n, \Delta(n))\bigr)$
for any $\Delta\colon \N \to \N$.
In particular,
$\NeciporukLB^{\BP}_{\ISA}(n) \in \Omega\bigl(\frac{n^2}{\log^2_2 n}\bigr)$.
\end{proposition}

\begin{proof}
  Let $\Delta\colon \N \to \N$.  Let $n \in \N, n \geq 32$.  Let $V_1,
  \ldots, V_p, U$ be a partition of $[n]$ such that $r_{V_i}(\isan) = 2^q$ for
  all $i \in [p]$ where $p, q \in \N_{>0}$ verify $p \geq \frac{1}{32} \cdot
  \frac{n}{\log_2 n }$ and $q \geq \frac{n}{16}$ as given by Lemma
  \ref{lem:ISA_partitioning}.  We have
\begin{align*}
& \NeciporukLB^{\LNBP_{\Delta(n)}}_{\ISA}(n)\\
\geq & \sum_{i = 1}^p b_{\LNBP_{\Delta(n)}}(r_{V_i}(\isan)) +
       b_{\LNBP_{\Delta(n)}}(r_U(\isan))\\
\geq & \sum_{i = 1}^p b_{\LNBP_{\Delta(n)}}(2^q) \\
\geq & \frac{1}{6} \cdot \frac{1}{32} \cdot \frac{n}{\log_2 n} \cdot
       \begin{cases}
       \max\bigl\{\frac{\frac{n}{16}}
		       {2^{\Delta(n)} (\log_2(\frac{n}{16}) - \Delta(n))},
		  \log_2(\frac{n}{16})\bigr\} &
       \text{if $4^{2^{\Delta(n)}} \leq 2^{\frac{n}{16}}$}\\
       \log_2(\frac{n}{16}) & \text{otherwise}
       \end{cases}\\
\geq & \frac{1}{192} \cdot \frac{n}{\log_2 n} \cdot \frac{1}{16} \cdot
       \begin{cases}
       \max\bigl\{\frac{n}{2^{\Delta(n)} (\log_2(n) - \Delta(n))},
		  \log_2 n\bigr\} & \text{if $2^{\Delta(n) + 5} \leq n$}\\
       \log_2 n & \text{otherwise}
       \end{cases} &
       \text{as $\log_2\Bigl(\frac{n}{16}\Bigr) \geq \frac{\log_2 n}{16}$}\\
\geq & \frac{c}{3072} \cdot \frac{n}{\log_2 n} \cdot
       \begin{cases}
       \max\bigl\{\frac{n}{2^{\Delta(n)} (\log_2(n) - \Delta(n))},
		  \log_2 n\bigr\} & \text{if $2^{\Delta(n) + 1} \leq n$}\\
       \log_2 n & \text{otherwise}
       \end{cases}& \text{for some $c$, see below}\\
= & \frac{c}{3072} \cdot
    \begin{cases}
    \max\bigl\{\frac{n^2}{2^{\Delta(n)} (\log_2(n) - \Delta(n)) \log_2 n},
	       n\bigr\} & \text{if $2^{\Delta(n) + 1} \leq n$}\\
    n & \text{otherwise}
    \end{cases}\\
= & \frac{c}{3072} \cdot \Xix_{\LNBP}(n, \Delta(n)) &\text{as desired}
\displaypunct{.}
\end{align*}

In order to show the inequality above it suffices to show that $\log_2 x
\geq \frac{cx}{2^\alpha(\log_2(x) - \alpha)}$ for
$x\in I=\ccinterval{2^{\alpha+1}}{2^{\alpha+5}}$, $\alpha\geq 0$ and some
constant $c$.

It suffices to show that the function
$f(x)=2^{\alpha}\log_2 (x)\log_2(\frac{x}{2^\alpha}) - cx$ is non-decreasing on
$I$. This concludes the claim as
$f(2^{\alpha+1})=2^{\alpha}(\alpha+1)-c2^{\alpha+1}\geq 0$ when $\alpha\geq 0$
and $c \leq \frac{1}{2}$.
To see this notice that the derivative of $f$ is
$2^\alpha(\frac{\log_2 x}{x \ln 2} + \frac{\log_2(\frac{x}{2^\alpha})}{x\ln
  2}) -c$ that has the same sign as $g(x)=2^\alpha\log_2(\frac{x^2}{2^\alpha})
- xc\ln 2$ for $x \in I$.

The derivative of $g$ is $\frac{2^{\alpha+1}}{x\ln 2} -c\ln 2$ that vanishes
for a value $x_0=\frac{2^{\alpha+1}}{c(\ln 2)^2}$. Assuming
$c\leq \frac{1}{2^4 (\ln 2)^2}$ we have $x_0 \geq 2^{\alpha+5}$ and the
derivative of $g$ is always non-negative on $I$.

We have $g(2^{\alpha+1})=2^\alpha(\alpha+2 - 2c\ln 2)$ which is non-negative as
soon as $c \leq \frac{1}{\ln 2}$. Hence $g$ is non-negative on $I$.

Hence taking $c=\frac{1}{2^4 (\ln 2)^2}$ yields the desired result.
\end{proof}

Now we show that for all $\delta \in \N$, $b_{\LNBP_\delta}$ is in fact an
asymptotically largest function in $\NeciporukSet_{\LNBP_\delta}$ (as well as
for $b_{\BP}$ and $\NeciporukSet_{\BP}$) and that the previous bound is in fact
also the asymptotically largest we may obtain, using the meta-results of
Section~\ref{sec:neciporuk}.
To do this, we appeal to our upper bound from Theorem~\ref{thm:all_uppers} on
the size of a $\delta$-LNBP (or a deterministic BP) computing $\isakl$ and apply
Lemma~\ref{lem:Neciporuk's_method_function_limitation_meta-result}.

\begin{proposition}
\label{ptn:Neciporuk's_method_function_limitation_LNBP_size}
\label{ptn:Neciporuk's_method_function_limitation_BP_size}
There exists a constant $c \in \R_{>0}$ verifying that for each $\delta \in \N$,
any $b \in \NeciporukSet_{\LNBP_\delta}$ is such that
$b(m) \leq c \cdot b_{\LNBP_\delta}(m)$ for all $m \in \N, m \geq 4$.
In particular, there exists a constant $c' \in \R_{>0}$ verifying that any
$b \in \NeciporukSet_{\BP}$ is such that
$b(m) \leq c' \cdot b_{\BP}(m)$ for all
$m \in \N, m \geq 4$.
\end{proposition}

\begin{proof}
Let $\delta \in \N$.
Let $g\colon \cointerval{1}{+\infty} \to \R_{\geq 0}$ be the non-decreasing
function defined by
\begin{align*}
g(x) = 13 \cdot 2^x \cdot
       \begin{cases}
       \max\bigl\{\frac{2^{x - \delta}}{x - \delta}, x\bigr\} &
	\text{if $\delta + 1 \leq x$}\\
       x & \text{otherwise}
       \end{cases}
\end{align*}
     
Theorem \ref{lem:ISA_LNBP_size_upper_bound} tells us that for all $k \in \N_{>0}$,
we have
\begin{align*}
\LNBP_\delta(\isakk)
& \leq \begin{cases}
       12 \cdot 2^k \max\bigl\{\frac{2^{k - \delta}}{k - \delta}, k\bigr\} +
       \frac{2^{2 k - \delta}}{k - \delta} & \text{if $\delta + 1 \leq k$}\\
       2^k (3 k + 1) + 2 \cdot 2^k & \text{otherwise}
       \end{cases}\\
& \leq \begin{cases}
       2^k \bigl(12 \max\bigl\{\frac{2^{k - \delta}}{k - \delta}, k\bigr\} +
       \frac{2^{k - \delta}}{k - \delta}\bigr) & \text{if $\delta + 1 \leq k$}\\
       6 \cdot 2^k k & \text{otherwise}
       \end{cases} & \text{as $k \geq 1$}\\
& \leq g(k)
\end{align*}
and moreover, $\frac{g(k + 1)}{g(k)} \leq 4$ for all $k \in \N_{>0}$. Indeed,
let $k \in \N_{>0}$, there are two cases to consider:
\begin{itemize}
\item
if $g(k+1)=13 \cdot 2^{k+1} (k+1)$ then notice that we always have $g(k)\geq 13
\cdot 2^k \cdot k$. Therefore we get
$\frac{g(k + 1)}{g(k)} \leq \frac{13 \cdot 2^{k + 1} (k + 1)}{13 \cdot 2^k k} =
 2 (1 + \frac{1}{k}) \leq 4$;
\item otherwise $g(k+1)= 13\cdot 2^{k+1} \frac{2^{k + 1 - \delta}}{k + 1 -
    \delta}$ and notice that either $g(k) \geq 13 \cdot 2^k \frac{2^{k - \delta}}{k -
    \delta}$ or $k=\delta$. If $k=\delta$ it is simple to check that $\frac{g(k +
    1)}{g(k)} \leq 4$ , otherwise we have $\frac{g(k +
    1)}{g(k)} \leq  \frac{13 \cdot 2^{k + 1} \frac{2^{k + 1 - \delta}}{k + 1 - \delta}}
      {13 \cdot 2^k \frac{2^{k - \delta}}{k - \delta}} =
 4 \cdot \frac{k - \delta}{k + 1 - \delta} \leq 4$.
\end{itemize}
Therefore, by Lemma
\ref{lem:Neciporuk's_method_function_limitation_meta-result}, any
$b\in \NeciporukSet_{\LNBP_\delta}$ verifies
\begin{align*}
b(m)
& \leq 4 \cdot \frac{g(\log_2\log_2 m)}{\log_2 m}\\
& = 4 \cdot
    \frac{13 \cdot 2^{\log_2\log_2 m }
	  \begin{cases}
	  \max\bigl\{\frac{2^{\log_2\log_2(m) - \delta}}
			  {\log_2\log_2(m) - \delta},
		     \log_2\log_2 m\bigr\} &
	  \text{if $\delta + 1 \leq \log_2\log_2 m$}\\
	  \log_2\log_2 m & \text{otherwise}
	  \end{cases}}{\log_2 m}\\
& = 52 \cdot
    \begin{cases}
    \max\bigl\{\frac{\log_2 m}{2^\delta (\log_2\log_2(m) - \delta)},
	       \log_2\log_2 m\bigr\} &
    \text{if $2^{2^{\delta + 1}} \leq m$}\\
    \log_2\log_2 m & \text{otherwise}
    \end{cases}\\
& \leq c \cdot b_{\LNBP_\delta}(m)
\end{align*}
for all $m \in \N, m \geq 4$, where $c \in \R_{>0}$ is a sufficiently large
constant.

In the case where $\delta =0$ notice that for $m\geq 4$ we have

$\log_2\log_2 m \leq d \cdot \frac{\log_2 x}{\log_2\log_2 x}$ for some
suitable constant $d$.

So we can also conclude that for any
$b \in \NeciporukSet_{\BP} = \NeciporukSet_{\LNBP_0}$, we have
\[
b(m) \leq 52 \cdot \max\Bigl\{\frac{\log_2 m}{\log_2\log_2 m},
			      \log_2\log_2 m\Bigr\} \leq
c' \cdot \frac{\log_2 m}{\log_2\log_2 m}
\]
for all $m \in \N, m \geq 4$, where $c' \in \R_{>0}$ is a sufficiently large
constant.
\end{proof}

Finally, using this and
Lemma~\ref{lem:Neciporuk's_method_limitation_complexity_lower_bound_meta-result},
we get the following result, showing that the asymptotically greatest lower
bound we may expect using \Neci's method for $\LNBP_\delta$ for any
$\delta \in \N$ is (asymptotically) equivalent to the lower bound for $\ISA$
given in Proposition \ref{ptn:ISA_LNBP_size_lower_bound}.

\begin{theorem}
\label{thm:Neciporuk's_method_limitation_LNBP_size_lower_bound}
\label{thm:Neciporuk's_method_limitation_BP_size_lower_bound}
For any family of Boolean functions $F = \{f_n\}_{n \in \N}$ and any
$\Delta\colon \N \to \N$,
$\NeciporukLB^{\LNBP_{\Delta(n)}}_F(n) \in
 \Omicron\bigl(\Xix_{\LNBP}(n, \Delta(n))\bigr)$.

In particular,
$\NeciporukLB^{\BP}_F(n) \in \Omicron\bigl(\frac{n^2}{\log^2_2 n}\bigr)$.
\end{theorem}

\begin{proof}
Let $\delta \in \N$. We aim at applying
Lemma~\ref{lem:Neciporuk's_method_limitation_complexity_lower_bound_meta-result}
which requires four hypotheses, (i) to (iv).

For (i), we set $h$ as $h_{\LNBP_\delta}(x)=
\begin{array}[t]{@{}r@{\,\,}l@{\,\,}l@{}}
\begin{cases}
\max\bigl\{\frac{\log_2 x}{2^\delta (\log_2\log_2(x) - \delta)},
	   \log_2\log_2 x\bigr\} & \text{if $2^{2^{\delta + 1}} \leq x$}\\
\log_2\log_2 x & \text{otherwise}
\displaypunct{.}
\end{cases}
\end{array}$ and $x_0 = 2^8$. One can verify that
$h$ is non-decreasing on $\cointerval{2^8}{+\infty}$.

For (ii), for all $x \in \cointerval{4}{+\infty}$, we have
$h(2^x) \geq \log_2\log_2(2^x) = \log_2 x$.

For (iii), for all $v, v' \in \N$ verifying $2^{2^v} \geq 2^8$ and $2^{2^{v'}}
\geq 2^8$, we need to show that $h(2^{2^v}) + h(2^{2^{v'}}) \leq h(2^{2^{v +
    v'}})$. There are three cases to consider.
\begin{itemize}
    \item
	If $h(2^{2^v}) = \log_2\log_2(2^{2^v}) = v$ and
	$h(2^{2^{v'}}) = \log_2\log_2(2^{2^{v'}}) = v'$, then
	$h(2^{2^v}) + h(2^{2^{v'}}) = v + v' =
	 \log_2\log_2(2^{2^{v + v'}}) \leq h(2^{2^{v + v'}})$.
       \item If $h(2^{2^v}) = \frac{\log_2(2^{2^v})}{2^\delta
           (\log_2\log_2(2^{2^v}) - \delta)} = \frac{2^{v - \delta}}{v -
           \delta} > v$ and $h(2^{2^{v'}}) = \log_2\log_2(2^{2^{v'}}) = v'$, then
         we necessarily have $v = \delta + \eta$ for some $\eta >0$.

	Notice that $\eta \geq 2$ because if $\eta=1$ then $v<2$, a contradiction.

	We conclude by showing that $h(2^{2^v}) + h(2^{2^{v'}}) =  \frac{2^{v - \delta}}{v -
           \delta} + v' \leq \frac{2^{\eta +v'}}{\eta+v'}\leq h(2^{2^{v+v'}})$.

        Only the first inequality is non immediate. To see it, consider the
        function $f(x)=2^{\eta+x} - (\frac{2^\eta}{\eta}+x)(\eta+x)$. A simple
        calculation shows that it is non-decreasing for $x\geq 2$ and
	$\eta\geq 2$.
	The inequality follows as $f(2)$ is non-negative when $\eta\geq 2$.
      \item In the remaining case $h(2^{2^v}) = \frac{2^{v - \delta}}{v -
          \delta}>v$ and $h(2^{2^{v'}}) = \frac{2^{v' - \delta}}{v' -
          \delta}>v'$. It implies that $v \geq \delta + 1$ and $v' \geq \delta +
        1$. Arguing as above we actually have $v \geq \delta + 2$
	and $v' \geq \delta + 2$, otherwise $v$ or $v'$ would be smaller than
        $2$. We then have:

	$h(2^{2^v}) + h(2^{2^{v'}}) =
	 \frac{2^{v - \delta}}{v - \delta} +
	 \frac{2^{v' - \delta}}{v' - \delta} \leq
	 \frac{2^{v + v' - 2 \delta}}{(v - \delta) (v' - \delta)} \leq
	 \frac{2^{v + v' - 2 \delta}}{v + v' - 2 \delta} \leq
	 h(2^{2^{v + v' - \delta}}) \leq h(2^{2^{v + v'}})$.

        The first and second inequality are because $x+y\leq xy$ when both $x$ and $y$ are
        greater than 2 (in the second case we use that $v-\delta \geq 2$ and
        $v'- \delta \geq 2$). The third one is by definition of $h$ and the last
        one by monotonicity of $h$.
\end{itemize}

For (iv), by Proposition~\ref{ptn:Neciporuk's_method_function_limitation_LNBP_size}, we
know that any $b \in \NeciporukSet_{\LNBP_\delta}$ is such that
$b(m) \leq \alpha\cdot h(m)$ for all $m \geq 4$.

We can therefore apply
Lemma~\ref{lem:Neciporuk's_method_limitation_complexity_lower_bound_meta-result}
with $x_0 = 2^8$ and get that for any family of Boolean
functions $F = \{f_n\}_{n \in \N}$ and all $n \in \N, n \geq 8$,
\begin{align*}
\NeciporukLB^{\LNBP_\delta}_F(n)
& \leq \alpha \cdot \bigl(4 + h(2^8)) \cdot
       \frac{n}{\log_2 n} \cdot h(2^n)\\
& = c \cdot \frac{n}{\log_2 n} \cdot
    \begin{cases}
    \max\bigl\{\frac{n}{2^\delta (\log_2(n) - \delta)}, \log_2 n\bigr\} &
    \text{if $2^{2^{\delta + 1}} \leq 2^n$}\\
    \log_2 n & \text{otherwise}
    \end{cases}\\
& = c \cdot
    \begin{cases}
    \max\bigl\{\frac{n^2}{2^\delta (\log_2(n) - \delta) \log_2 n}, n\bigr\} &
    \text{if $2^{\delta + 1} \leq n$}\\
    n & \text{otherwise}
    \end{cases}\\
& = c \cdot \Xix_{\LNBP}(n, \delta)
\displaypunct{.}
\end{align*}

Thus, since this holds for all $\delta \in \N$, we get the desired result.
\end{proof}


\section{Deterministic and Limited Nondeterministic Formulas}
\label{sec:Limited_nondeterministic_formulas}
In this section, we focus on the model of Boolean binary formulas and its
limited nondeterministic variant. $\BF$ is one of the two measures that were
considered in \Neci's original article~\cite{ne66} who gave a
$O(\frac{n^2}{\log_2 n})$ lower bound for this complexity measure. If the
model is restricted to the case of binary formulas where only $2$-ary AND and
OR gates can be used, stronger lower bounds can be proven, the best known for
instance being almost cubic and due to H\r{a}stad (see \cite[Theorem
6.15]{ju12}).  Just as in
Section~\ref{sec:Limited_nondeterministic_branching_programs}, results for the
Ne\v{c}iporuk method for binary formulas are known (see for instance
\cite[Chapter 8, Section 7]{we87}), but we do not know about any
attempt to consider the method in its full generality: an approach that
would explicitly try to find the best Ne\v{c}iporuk function rather than just
giving one, as done in \cite{alzw89} for the case of BPs.

Concerning limited nondeterministic binary formulas, \Neci's lower bound method
never seems to have been applied to the associated complexity measure, at least
in a direct combinatorial sense that excludes Klauck's communication complexity
formulation of the method~\cite{kl07}.

For all $\delta \in \N$, let us define the function
$b_{\LNBF_\delta}\colon \N_{>0} \to \N$ given by 
\[
b_{\LNBF_\delta}(m)  =
\begin{cases}
\ceiling{\frac{1}{4} \max\bigl\{\frac{\log_2 m}{2^\delta},
				      \log_2\log_2 m\bigr\}} &
    \text{if $m \geq 4$}\\
0 & \text{otherwise}
\end{cases}
\]
for all $m \in \N_{>0}$.
We denote by $b_{\BF}$ the case of $b_{\LNBF_0}$.

We first prove that $b_{\BF} \in \NeciporukSet_{\BF}$ and
$b_{\LNBF_\delta} \in \NeciporukSet_{\LNBF_\delta}$. This is similar to the
limited nondeterministic branching program case.

\begin{proposition}
\label{ptn:Neciporuk_LNBF_size_lower_bound}
\label{ptn:Neciporuk_BF_size_lower_bound}
$b_{\LNBF_\delta} \in \NeciporukSet_{\LNBF_\delta}$ for all $\delta \in \N$.
In particular, $b_{\BF} \in \NeciporukSet_{\BF}$.
\end{proposition}

\begin{proof}
Let $\delta \in \N$. It is fairly obvious that $b_{\LNBF_\delta}$ is
non-decreasing.

Let $f$ be a $n$-ary boolean function on $V$ and let $V_1, \ldots, V_p$ a
partition of $V$.  Let $\phi$ be a Boolean $\delta$-LNBF computing $f$ and let
$g$ be the $(n+\delta)$-ary Boolean function computed by $\phi$ when
considering the $\delta$ nondeterministic bits as regular input variables.

For all $i \in [p]$ we will denote by $s_i \in \N$ the number of leaves in
$\phi$ labelled by literals whose variable indices are in $V_i$, as well as
$q \in \N$ the number of leaves in $\phi$ labelled by literals whose variable
indices are not in $V$. It is clear that
$|\phi| = \sum_{i = 1}^p s_i + q \geq \sum_{i = 1}^p s_i$.
To conclude it remains to show that $s_i \geq b_{\LNBF_\delta}(r_{V_i}(f))$ for all
$i \in [p]$.

Fix $i \in [p]$. The claim is obvious if $r_{V_i}(f) \leq 3$ hence we assume
$r_{V_i}(f) \geq 4$. Let $V_i'$ be the subset of $V_i$ containing all indices
of variables on which $f$ depends. Then, by Lemma
\ref{lem:Number_of_subfunctions_with_variable_independence}, $r_{V_i}(f) =
r_{V_i'}(f)$. Moreover for each $l \in V_i'$, $\phi$ contains at least one leaf
labelled by $l$. By Lemma~\ref{lem:Number_of_subfunctions_upper_bound}, it
follows that $r_{V_i}(f) = r_{V_i'}(f) \leq 2^{2^{\card{V_i'}}} \leq
2^{2^{s_i}}$. So we can conclude that $s_i \geq \log_2\log_2(r_{V_i}(f))$.

If $r_{V_i}(f) \leq 2^{2^{\delta + 1}}$, we have
\[
\ceiling{\frac{1}{4} \cdot \frac{\log_2(r_{V_i}(f))}{2^\delta}} \leq
\ceiling{\frac{1}{4} \cdot
	 \frac{\log_2\bigl(2^{2^{\delta + 1}}\bigr)}{2^\delta}} =
1 \leq
\ceiling{\frac{1}{4} \cdot \log_2\log_2(r_{V_i}(f))}
\displaypunct{.}
\]
and therefore $s_i \geq  b_{\LNBF_\delta}(r_{V_i}(f))$.

It remains to consider the case where $r_{V_i}(f) > 2^{2^{\delta + 1}}$.
Notice that this implies $s_i > 0$, as $2^{2^{s_i}} \geq r_{V_i}(f)$.

This part of the proof is taken from classical references, e.g.  \cite[Proof of
Theorem 7.1]{we87} or \cite[Proof of Theorem 6.16]{ju12}. We denote by $T_i$
the sub-tree of $\phi$ consisting of all paths from a leaf with a label in
$V_i$ to the root of $\phi$. This tree has nodes of fan-in $0$, $1$ or $2$ and
is non-empty since $s_i > 0$. Let $W_i$ be the set of nodes of $T_i$ that have
fan-in $2$ and notice that $\card{W_i} \leq s_i - 1$. Let $P_i$ be the set of
paths in $T_i$ starting from a leaf or a node in $W_i$ and ending in a node in
$W_i$ or in the root of $T_i$ and containing no node in $W_i$ as inner node.
Notice that $\card{P_i} \leq 2 \card{W_i} + 1 \leq 2 s_i$.

For any partial assignment $\rho \in \{0, 1\}^{V \setminus V_i \cup [\delta]}$,
we obtain a formula $\phi|_\rho$ of size $s_i$ computing $g|_\rho$ by replacing
each variable in $V \setminus V_i \cup [\delta]$ by the appropriate constant
given by $\rho$. This assignment induces that any part of $\phi|_\rho$
corresponding to a path $p$ in $P_i$, either computes a constant function, or is
the identity or negates its input. Reciprocally any of these four choices on $p$
induces a subfunction of $g$. Hence we have
$r_{V_i}(g) \leq 4^{\card{P_i}}\leq 2^{4s_i}$.

As $g$ is a proof-checker function for $f$, from
Lemma~\ref{lem:Number_of_subfunctions_proof_checker_upper_bound} it follows
that $r_{V_i}(g) \geq r_{V_i}(f)^{\frac{1}{2^\delta}}$, therefore
\[
s_i \geq \frac{1}{4} \log_2(r_{V_i}(g)) \geq
\frac{1}{4} \log_2\Bigl(r_{V_i}(f)^{\frac{1}{2^\delta}}\Bigr) =
\frac{1}{4} \cdot \frac{\log_2(r_{V_i}(f))}{2^\delta}
\displaypunct{.}
\]
Altogether, we have
\[
s_i \geq \frac{1}{4} \max\Bigl\{\frac{\log_2(r_{V_i}(f))}{2^\delta},
				\log_2\log_2(r_{V_i}(f))\Bigr\}
\displaypunct{,}
\]
which implies that $s_i \geq b_{\LNBF_\delta}(r_{V_i}(f))$ as $s_i$ is integral.

In conclusion for any $n$-ary Boolean function $f$ on $V$ and any partition $V_1, \ldots,
V_p$ of $V$, it holds that $\LNBF_\delta(f) \geq \sum_{i = 1}^p
b_{\LNBF_\delta}(r_{V_i}(f))$, hence $b_{\LNBF_\delta} \in
\NeciporukSet_{\LNBF_\delta}$.
\end{proof}

Using this and Lemma \ref{lem:ISA_partitioning}, we can immediately derive the
following asymptotic lower bound on $\NeciporukLB^{\LNBF_{\Delta(n)}}_{\ISA}$.

\begin{proposition}
\label{ptn:ISA_LNBF_size_lower_bound}
\label{ptn:ISA_BF_size_lower_bound}
$\NeciporukLB^{\LNBF_{\Delta(n)}}_{\ISA}(n) \in
 \Omega\bigl(\max\bigl\{\frac{n^2}{2^{\Delta(n)} \log_2 n}, n\bigr\}\bigr)$
for any $\Delta\colon \N \to \N$.
In particular,
$\NeciporukLB^{\BF}_{\ISA}(n) \in \Omega\bigl(\frac{n^2}{\log_2 n}\bigr)$.
\end{proposition}

\begin{proof}
Let $\Delta\colon \N \to \N$.
Let $n \in \N, n \geq 32$.
Let $V_1, \ldots, V_p, U$ be a partition of $V$ such that
$r_{V_i}(\isan) = 2^q$ for all $i \in [p]$ where $p, q \in \N_{>0}$ verify
$p \geq \frac{1}{32} \cdot \frac{n}{\log_2 n}$ and $q \geq \frac{n}{16}$ as
given by Lemma~\ref{lem:ISA_partitioning}.
We have
\begin{align*}
\NeciporukLB^{\LNBF_{\Delta(n)}}_{\ISA}(n)
&\geq \sum_{i = 1}^p b_{\LNBF_{\Delta(n)}}(r_{V_i}(\isan)) +
       b_{\LNBF_{\Delta(n)}}(r_U(\isan))\\
& \geq \sum_{i = 1}^p \frac{1}{4} \cdot 
		      \max\Bigl\{\frac{\log_2(2^q)}{2^{\Delta(n)}},
				 \log_2\log_2(2^q)\Bigr\}\\
& = p \cdot \frac{1}{4} \cdot
    \max\Bigl\{\frac{q}{2^{\Delta(n)}}, \log_2 q\Bigr\}\\
& \geq c_1 \cdot \frac{n}{\log_2 n} \cdot
       \max\Bigl\{\frac{n}{16 \cdot 2^{\Delta(n)}},
		  \log_2\Bigl(\frac{n}{16}\Bigr)\Bigr\}\\
& \geq c_2 \cdot
    \max\Bigl\{\frac{n^2}{2^{\Delta(n)} \log_2 n}, n\Bigr\} \text{~~~because }
    n\geq 32 \text{ implies }\log_2(\frac{n}{16}) \geq \frac{\log_2 n}{16}
\displaypunct{.}
\end{align*}
\end{proof}

We now show that for all $\delta \in \N$,
$b_{\LNBF_\delta}$ is in fact an asymptotically largest function in
$\NeciporukSet_{\LNBF_\delta}$. To this end, we appeal to the
upper bound on $\LNBF_\delta(\isakl)$ from Theorem~\ref{thm:all_uppers} and
apply Lemma~\ref{lem:Neciporuk's_method_function_limitation_meta-result}.

\begin{proposition}
\label{ptn:Neciporuk's_method_function_limitation_LNBF_size}
\label{ptn:Neciporuk's_method_function_limitation_BF_size}
There exists a constant $c \in \R_{>0}$ verifying that for each $\delta \in \N$,
any $b \in \NeciporukSet_{\LNBF_\delta}$ is such that
$b(m) \leq c \cdot b_{\LNBF_\delta}(m)$ for all $m \in \N, m \geq 4$.
\end{proposition}

\begin{proof}
  Fix $\delta \in \N$.  Let $g\colon \cointerval{1}{+\infty} \to \R_{\geq 0}$
  be the non-decreasing function defined by $g(x) = 15 \cdot 2^x \cdot
  \max\{2^{x - \delta}, x\}$.

  Notice that $\frac{g(k + 1)}{g(k)} \leq 4$ for
  all $k \in \N_{>0}$.

 Moreover, from Theorem~\ref{lem:ISA_LNBF_size_upper_bound} we have:
\[
\LNBF_\delta(\isakk)
\leq 12 \cdot 2^k \cdot \max\{2^{k - \delta}, k\} + 3 \cdot 2^k
\leq g(k).
\]
Therefore, by Lemma
\ref{lem:Neciporuk's_method_function_limitation_meta-result}, any
$b\in \NeciporukSet_{\LNBF_\delta}$ verifies
\begin{align*}
b(m)
& \leq 4 \cdot \frac{g(\log_2\log_2 m)}{\log_2 m}\\
& = 4 \cdot
    \frac{15 \cdot 2^{\log_2\log_2 m} \cdot
	  \max\{2^{\log_2\log_2(m) - \delta}, \log_2\log_2 m\}}
	 {\log_2 m}\\
& = 60 \cdot \max\Bigl\{\frac{\log_2 m}{2^\delta}, \log_2\log_2 m\Bigr\}
\end{align*}
for all $m \geq 4$.
\end{proof}

Finally, using this and
Lemma~\ref{lem:Neciporuk's_method_limitation_complexity_lower_bound_meta-result},
we show that asymptotically the greatest lower
bound we may expect using \Neci's method is the one obtained for $\ISA$
in Proposition~\ref{ptn:ISA_LNBF_size_lower_bound}.

\begin{theorem}
\label{thm:Neciporuk's_method_limitation_LNBF_size_lower_bound}
\label{thm:Neciporuk's_method_limitation_BF_size_lower_bound}
For any family of Boolean functions $F = \{f_n\}_{n \in \N}$ and any
$\Delta\colon \N \to \N$,
$\NeciporukLB^{\LNBF_{\Delta(n)}}_F(n) \in
 \Omicron\bigl(\max\bigl\{\frac{n^2}{2^{\Delta(n)} \log_2 n},
 n\bigr\}\bigr)$.

In particular,
$\NeciporukLB^{\BF}_F(n) \in \Omicron\bigl(\frac{n^2}{\log_2 n}\bigr)$.
\end{theorem}

\begin{proof}
Fix $\delta \in \N$. We aim at applying
Lemma~\ref{lem:Neciporuk's_method_limitation_complexity_lower_bound_meta-result}
which requires four hypotheses.

For (i), let $h\colon \cointerval{4}{+\infty} \to \R_{\geq 0}$ be the function
defined by
$$h(x) = \max\bigl\{\frac{\log_2 x}{2^\delta}, \allowbreak
		   \log_2\log_2 x\bigr\}$$
and $x_0 = 2^8$; as required, $h$ is
non-decreasing on $\cointerval{2^8}{+\infty}$.

For (ii), for all $x \in \cointerval{4}{+\infty}$, we have
$h(2^x) \geq \log_2\log_2(2^x) = \log_2 x$.

For (iii), for all $v, v' \in \N$ verifying  $2^{2^v} \geq 2^8$ and
$2^{2^{v'}} \geq 2^8$, we have
$h(2^{2^v}) + h(2^{2^{v'}}) \leq h(2^{2^{v + v'}})$. Indeed, let $v, v' \in \N$
such that  $2^{2^v} \geq 2^8$ and $2^{2^{v'}} \geq 2^8$, there are two cases to
consider.
\begin{itemize}
\item
If $h(2^{2^v}) = \log_2\log_2(2^{2^v}) = v$ and
$h(2^{2^{v'}}) = \log_2\log_2(2^{2^{v'}}) = v'$, then
$h(2^{2^v}) + h(2^{2^{v'}}) = v + v' = \log_2\log_2(2^{2^{v + v'}}) \leq
 h(2^{2^{v + v'}})$.
\item
Otherwise, there is at least one $w \in \{v, v'\}$ such that
$h(2^{2^w}) = \frac{\log_2(2^{2^w})}{2^\delta} = 2^{w - \delta}$: assume without
loss of generality that it is $v$.
Then, since $2^{2^v} \geq 16$, we have $v \geq 2$, so by hypothesis, it follows
that $2^{v - \delta} > v \geq 2$. Moreover,
$h(2^{2^{v'}}) = \max\{2^{v' - \delta}, v'\} \leq 2^{v'}$, so using our usual
observation about the relationship between the sum and the product of two real
numbers greater than or equal to $2$, we get $h(2^{2^v}) + h(2^{2^{v'}}) \leq 2^{v - \delta} + 2^{v'} \leq
 2^{v + v' - \delta} = \frac{\log_2(2^{2^{v + v'}})}{2^\delta} \leq
 h(2^{2^{v + v'}})$.
\end{itemize}

For (iv), by Proposition~\ref{ptn:Neciporuk's_method_function_limitation_LNBF_size}, we
know that any $b \in \NeciporukSet_{\LNBF_\delta}$ is such that
$b(m) \leq \alpha\cdot h(m)$ for all $m \geq 4$.

Therefore, by
Lemma~\ref{lem:Neciporuk's_method_limitation_complexity_lower_bound_meta-result},
for any family of Boolean functions $F = \{f_n\}_{n \in \N}$ and all $n\geq 8$, we
have:
\begin{align*}
\NeciporukLB^{\LNBF_\delta}_F(n)
& \leq \alpha \cdot \bigl(4 + h(\floor{2^8})\bigr) \cdot
       \frac{n}{\log_2 n} \cdot h(2^n)\\
& \leq c \cdot \frac{n}{\log_2 n} \cdot
    \max\Bigl\{\frac{n}{2^\delta}, \log_2 n\Bigr\}\\
& = c \cdot \max\Bigl\{\frac{n^2}{2^\delta \log_2 n}, n\Bigr\}
\displaypunct{.}
\end{align*}
\end{proof}


\section{Conclusion}
We have proposed a general interpretation of what it means to say ``the method
of \Neci''. We have applied the method to several complexity measures, as
reported in Table~\ref{tbl:Summary_lower_bounds}, and shown in particular that
the limitations of the method are very much determined by the complexity of the
Indirect Storage Access function under each measure, at least for those we
studied in this paper.
Note that our focus was not on optimizing the constant factors in
the bounds obtained, most of which can certainly be improved.

Our abstract definition of a \Neci function is inspired by Alon and
Zwick~\cite{alzw89}. It has the benefit of not specifying the way in which such
a function is obtained, be it some ``semantic count'' of the number of different
Boolean functions computable with a given cost, some ``syntactic count'' of the
number of different devices of that cost as done usually, or any other
technique.
While in the literature, ``\Neci-style theorems'' refer to giving an explicit
\Neci function as defined in step 1 in
Definition~\ref{defi-neci}~\cite{we87,ju12,alzw89}, it is natural to ask
whether we could even further twist the definition of a \Neci function to get
more out of the method. 

Looking at our meta-results and how we draw the limitation results for \Neci's
method used for a specific measure $\M$, namely using an upper bound on the
$\isakk$ function for all $k \in \N_{>0}$, we observe that the main weakness of
the method is that a \Neci function for $\M$ should verify the conditions
presented in step 1 of Definition \ref{defi-neci} for \emph{any} Boolean
function. The natural question is therefore whether restricting the class of
Boolean functions for which these conditions should be verified by a \Neci
function for $\M$ would allow to get stronger \Neci functions (and thus, lower
bounds) for $\M$ for this specific class of Boolean functions. This seems to be
an interesting question to us, but is not treated in this paper.

\bibliographystyle{splncs_srt}
\bibliography{all}

\begin{thebibliography}{10}

\bibitem{alzw89}
Alon, N., Zwick, U.:
\newblock On {N}e{\u{c}}iporuk's theorem for branching programs.
\newblock Theor. Comput. Sci. \textbf{64}(3) (1989)  331--342

\bibitem{bosi90}
Boppana, R.B., Sipser, M.:
\newblock The complexity of finite functions.
\newblock In {van Leeuwen}, J., ed.: Handbook of Theoretical Computer Science.
  Volume~A.
\newblock Elsevier (1990)  757--804

\bibitem{co66}
Cobham, A.:
\newblock The recognition problem for the set of perfect squares.
\newblock In: 7th Annual Symposium on Switching and Automata Theory, Berkeley,
  California, USA, October 23-25, 1966. (1966)  78--87

\bibitem{golemu96}
Goldsmith, J., Levy, M.A., Mundhenk, M.:
\newblock Limited nondeterminism.
\newblock {SIGACT} News \textbf{27}(2) (1996)  20--29

\bibitem{hs03:limited-advice}
Hromkovic, J., Schnitger, G.:
\newblock Nondeterministic communication with a limited number of advice bits.
\newblock {SIAM} J. Comput. \textbf{33}(1) (2003)  43--68

\bibitem{juk01}
Jukna, S.:
\newblock Extremal combinatorics - with applications in computer science.
\newblock Springer (2001)

\bibitem{ju12}
Jukna, S.:
\newblock Boolean function complexity: advances and frontiers. Volume~27.
\newblock springer (2012)

\bibitem{kawi93}
Karchmer, M., Wigderson, A.:
\newblock On span programs.
\newblock In: Proceedings 8th Structure in Complexity Theory, IEEE Computer
  Society Press (1993)  102--111

\bibitem{kl07}
Klauck, H.:
\newblock One-way communication complexity and the {N}e{\u{c}}iporuk lower
  bound on formula size.
\newblock SIAM J. Comput. \textbf{37}(2) (2007)  552--583

\bibitem{lup58}
Lupanov, O.B.:
\newblock A method of circuit synthesis.
\newblock Izvestia V.U.Z. Radiofizika \textbf{1} (1958)  120--140

\bibitem{ma76}
Masek, W.:
\newblock A fast algorithm for string editing problem and decision graph
  complexity.
\newblock Technical report, Massachussetts Institute of Technology (1976)

\bibitem{ne62}
Ne{\u{c}}iporuk, {\`{E}}.:
\newblock On the complexity of schemes in some bases containing nontrivial
  elements with zero weights.
\newblock Problemy Kibernetiki \textbf{8} (1962)  123--160 (in Russian).

\bibitem{ne66}
Ne{\u{c}}iporuk, {\`{E}}.:
\newblock On a boolean function.
\newblock Doklady of the Academy of the USSR \textbf{169}(4) (1966)  765--766
  Translation: \emph{Soviet Math. Doklady} 7:4, pp.\ 999-1000.

\bibitem{pau78}
Paul, W.:
\newblock Komplexit{{{{{{{{{{{\"a}}}}}}}}}}}tstheorie.
\newblock Leitf{{{{{{{{{{{\"a}}}}}}}}}}}den der angewandten Mathematik und
  Mechanik LAMM. Teubner Studienb{{{{{{{{{{{\"u}}}}}}}}}}}cher (1978)

\bibitem{pu87}
Pudl{\'a}k, P.:
\newblock The hierarchy of boolean circuits.
\newblock Computers and artificial intelligence \textbf{6}(5) (1987)  449--468

\bibitem{ra91}
Razborov, A.:
\newblock Lower bounds for deterministic and nondeterministic branching
  programs.
\newblock In: 8th Internat. Symp. on Fundamentals of Computation Theory. (1991)
   47--60

\bibitem{sav76}
Savage, J.E.:
\newblock The Complexity of Computing.
\newblock John Wily, New York (1976)

\bibitem{we87}
Wegener, I.:
\newblock The Complexity of Boolean Functions.
\newblock Wiley-Teubner series in computer science. B. G. Teubner \& John
  Wiley, Stuttgart (1987)

\bibitem{we00}
Wegener, I.:
\newblock Branching Programs and Binary Decision Diagrams: Theory and
  Applications.
\newblock SIAM Monographs on Discrete Mathematics and Applications. SIAM,
  Philadelphia (2000)

\end{thebibliography}
\end{document}